\documentclass[11pt]{article}
\usepackage{authblk}
\usepackage[margin=1in]{geometry}

\usepackage{amsmath,amsthm,amsfonts,amssymb,latexsym,verbatim,bbm}
\usepackage{subcaption}
\usepackage[shortlabels]{enumitem}
\usepackage{afterpage}
\usepackage{mathtools}
\usepackage{mathrsfs}
\usepackage[bookmarksnumbered,linktocpage,hypertexnames=false,colorlinks=true,linkcolor=blue,urlcolor=magenta,citecolor=magenta,anchorcolor=green,breaklinks=true]{hyperref}

\usepackage{bm,xparse,xcolor,tikz}
\usepackage[nodayofweek]{datetime}  
\usepackage[capitalise]{cleveref}

\newcommand{\ket}[1]{| #1 \rangle}
\newcommand{\bra}[1]{\langle #1 |}
\newcommand{\braket}[2]{\langle #1|#2\rangle}
\newcommand{\ketbra}[2]{|#1\rangle\!\langle#2|}

\newcommand{\tr}{{\mathrm tr}}

\newcommand{\pseudo}{\tilde{\mathop{\mathbb{E}}}_{\mathrm{S}}}
\newcommand{\expect}{\mathop{\mathbb{E}}}

\newcommand{\p}{\mathsf{p}}
\newcommand{\qmodel}{\mathrm{Q}}
\newcommand{\cmodel}{\mathrm{C}}

\newcommand{\dec}{\mathsf{Dec}}
\newcommand{\enc}{\mathsf{Enc}}
\newcommand{\gen}{\mathsf{Gen}}
\newcommand{\eval}{\mathsf{Eval}}
\newcommand{\sk}{\mathsf{sk}}

\newcommand{\negl}{\mathrm{negl}}

\numberwithin{equation}{section}
\theoremstyle{definition}
\newtheorem{theorem}{Theorem}[section]

\newtheorem{lemma}[theorem]{Lemma}

\newtheorem{definition}[theorem]{Definition}

\newtheorem{remark}[theorem]{Remark}

\newtheorem{claim}{Claim}

\newcommand{\R}{\mathbb{R}}
\newcommand{\C}{\mathbb{C}}

\newcommand{\N}{\mathbb{N}}

\newcommand{\Id}{\mathbbm{1}}

\newcommand{\mcA}{\mathcal{A}}
\newcommand{\mcB}{\mathcal{B}}
\newcommand{\mcC}{\mathcal{C}}
\newcommand{\mcD}{\mathcal{D}}

\newcommand{\mcF}{\mathcal{F}}

\newcommand{\mcH}{\mathcal{H}}
\newcommand{\mcI}{\mathcal{I}}

\newcommand{\mcL}{\mathcal{L}}

\newcommand{\mcN}{\mathcal{N}}

\newcommand{\mcP}{\mathcal{P}}

\newcommand{\mcS}{\mathcal{S}}

\newcommand{\mcX}{\mathcal{X}}
\newcommand{\mcY}{\mathcal{Y}}

\newcommand{\mfA}{\mathfrak{A}}

\newcommand{\mpC}{\mathsf{C}}

\title{Self-testing in the compiled setting via tilted-CHSH inequalities}

 \author[1]{Arthur Mehta\thanks{amehta2@uottawa.ca}}
 \author[1]{Connor Paddock\thanks{cpaulpad@uottawa.ca}}
 \author[2,3,4]{Lewis Wooltorton\thanks{lewis.wooltorton@ens-lyon.fr}}

\affil[1]{Department of Mathematics and Statistics, University of Ottawa, Canada}
\affil[2]{Department of Mathematics, University of York, UK}
\affil[3]{Quantum Engineering Centre for Doctoral Training, H. H. Wills Physics Laboratory and Department of Electrical \& Electronic Engineering, University of Bristol, UK}
\affil[4]{Inria, ENS de Lyon, LIP, 46 Allee d’Italie, 69364 Lyon Cedex 07, France}

\date{}

\begin{document}
\maketitle

\begin{abstract}
This work investigates the family of extended tilted-CHSH inequalities in the single-prover cryptographic compiled setting. In particular, we show that a quantum polynomial-time prover can violate these Bell inequalities by at most negligibly more than the violation achieved by two non-communicating quantum provers. To obtain this result, we extend a sum-of-squares technique to monomials with arbitrarily high degree in the Bob operators and degree at most one in the Alice operators. We also introduce a notion of partial self-testing for the compiled setting, which resembles a weaker form of self-testing in the bipartite setting. As opposed to certifying the full model, partial self-testing attempts to certify the reduced states and measurements on separate subsystems. In the compiled setting, this is akin to the states after the first round of interaction and measurements made on that state. Lastly, we show that the extended tilted-CHSH inequalities satisfy this notion of a compiled self-test.
\end{abstract}

\section{Introduction}

In a bipartite Bell scenario, two non-communicating provers receive inputs $x$ and $y$ and reply with outputs $a$ and $b$ to a verifier. The collection of probabilities of observing outcomes $(a,b)$ given $(x,y)$ determines a correlation $\mathsf{p}=\{p(a,b|x,y)\}$. Bell's celebrated theorem implies that if the provers are permitted to share an entangled quantum state and make local quantum measurements, called a bipartite (quantum) model, then certain correlations have no realization by a classical (or local hidden-variable) model \cite{Bell64}. The distinction between quantum and classical correlations is often explored through Bell inequalities. A Bell inequality is a linear inequality on the set of correlations which is satisfied by all classical correlations. Hence, these inequalities can be violated by certain models using quantum entanglement, realizing correlations that are not classical. The quantum value of a Bell inequality refers to the largest violation achievable by a bipartite (quantum) model. A prominent example is the Clauser-Horne-Shimony-Holt (CHSH) inequality, where the classical bound is $2$, but the quantum value is $2\sqrt{2}$ \cite{CHSH}.

Due to their ability to witness these non-classical effects, Bell inequality violations play a major role in areas like device-independent cryptography~\cite{ABGMPS,ColbeckThesis,PAMBMMOHLMM,PABGMS,Sca12}, protocols for verifiable delegated quantum computation~\cite{RUV13,Gri19,coladangelo2024verifier}, and in the study of multiprover interactive proofs (MIPs) and the variant MIP$^*$ with entangled provers \cite{CHTW04}, also called nonlocal games. Many of the key applications of Bell inequalities rely on a remarkable property known as \emph{self-testing} \cite{mayers2004self,MYS12,YangSelfTest,SupicSelfTest}. Informally, a Bell inequality is a self-test for an ideal bipartite (quantum) model $\qmodel$ if there exist local isometries which transform any employed bipartite model $\qmodel'$ achieving maximum Bell violation into the ideal model $\qmodel$. It is well-known that the CHSH inequality is a self-test for the bipartite model employing a maximally entangled state on two qubits, along with the Pauli $\sigma_x$ and $\sigma_z$ measurements, among others \cite{mayers2004self}. Another prominent example is the family of tilted-CHSH inequalities~\cite{AcinRandomnessNonlocality,YangSelfTest,BampsPironio}, which self-test partially entangled two-qubit states, and were integral in the work of Coladangelo, Goh, and Scarani who employed them as part of a protocol to self-test any pure bipartite entangled state \cite{Coladangelo_2017}. 

Despite the enormous success of self-testing, a practical drawback is the requirement of multiple non-communicating quantum provers. Recently, a number of cryptographic approaches have been proposed that replace the non-communication assumption with computational assumptions \cite{KLVY23, MV21, Kahanamoku_Meyer_2022}. This makes the setting more practical by having a single quantum prover, rather than multiple. One new and prominent approach is the Kalai-Lombardi-Vaikuntanathan-Yang (KLVY) compilation procedure introduced in \cite{KLVY23}, which transforms a $2$-prover $1$-round Bell scenario into a $1$-prover $2$-round scenario with a single computationally bounded prover. The core ingredient in the KLVY compilation procedure is quantum homomorphic encryption (QHE), which emulates, to a certain extent, the non-communication between the rounds of interaction. In the compiled game, the inputs to the prover happens sequentially. In the first round, the prover obtains an encryption $\chi$ of the input $x$ from the verifier. Without breaking the security, the prover cannot distinguish between encryptions of different inputs. The prover performs a polynomial time quantum circuit on $\chi$, and then returns an output $\alpha$ to the verifier. In the second round, the information about $x$ has already been ``hidden'' from the prover, so the verifier can send input $y$ in the plain (i.e.~unencrypted) to the prover, upon which the prover can perform a measurement and return outcome $b$ to the verifier. The verifier checks for a Bell inequality violation (across many such interactions) using the values of $x$, the decryption of $\alpha$, along with $(y,b)$. QHE has two key features that makes this resemble the bipartite setting. Firstly, it allows the first round quantum prover to perform measurements as they would have in the bipartite setting, without knowing the input. Secondly, the encryption ensures that no classical polynomial-time prover can violate a Bell inequality by more than an negligible amount (see \cref{sec:comp} details). Both of these facts are non-trivial and were the subject of \cite{KLVY23}.

In a follow-up work, Natarajan and Zhang showed that the maximal quantum violation of the CHSH inequality in the compiled setting is bounded by the maximal violation in the bipartite setting, up to negligible factors in the security parameter \cite{NZ23}. Subsequent works have analyzed the quantum soundness of the KLVY compilation procedure for other multiprover scenarious, including all 2-player XOR nonlocal games \cite{Cui24}, Bell inequalities tailored to maximally entangled bipartite states~\cite{Baroni24}, delegated quantum computation with a single-device \cite{NZ23, MNZ24}, and even in the study of contextuality \cite{Aro24}. Despite these advancements, many results have yet to be reproduced in the compiled setting. Our work takes another step in growing the list of protocols that will function as desired in the compiled setting.

\paragraph*{Upper bounding compiled Bell violations}

As mentioned, the compiled value of a Bell inequality is always at least the quantum value. This is because any bipartite (quantum) model can be implemented with homomorphic encryption via a \emph{correctness} property of the QHE scheme used in the procedure. On the other hand, establishing upper bounds on the largest violation possible in the compiled setting is challenging, as general techniques for bounding these violations depend on the spatial separation between the two provers. Nonetheless, upper bounds on the violations of a certain Bell inequalities in the compiled setting can be verified using the sum-of-squares (SOS) technique~\cite{NZ23,Cui24,Baroni24}. The SOS approach is a powerful method and has been used extensively to upper bound Bell inequality violations and the values of nonlocal games in the bipartite setting. Informally, this technique relates the maximum compiled value $\eta$, of a Bell functional $I$, to a decomposition of the Bell operator or Bell polynomial $S$ as a sum of Hermitian squares, $\eta \Id - S = \sum_i P^\dagger_iP_i$. Before our work, progress was made on realizing this approach in the compiled setting, however, there were some limitations. In particular, it was required that the polynomials $P_i$ involved in the decomposition were at most degree two in both Alice's and Bob's observables, restricting the technique to Bell inequalities with an SOS decomposition of this form; this excludes, for example, the family of tilted-CHSH inequalities. 

Our first result extends the SOS technique to a larger family of Bell polynomials. More specifically, we extend the pseudo-expectation techniques in \cite{NZ23,Cui24} to allow for evaluations on polynomial terms $P_i$ that consist of arbitrary monomials in the algebra generated by Bob's observables. In \cref{thm:ExtendedSOS} we prove that an extended pseudo-expectation will be positive on the corresponding Hermitian square $P_{i}^{\dagger}P_{i}$ for any such term $P_i$. Consequently, we show that for any Bell inequality with an SOS decomposition in which $P_i$ are of the form $P_i = \sum_{j}\gamma_{j}(A_{x})^{k_{j}}w_{j}(B)$ for some $\gamma_{j} \in \mathbb{C}$, $k_{j} \in \{0,1\}$ and $w_{j}(B)$ being arbitrary monomials in Bob's observables, $\eta$ is an upper-bound on the maximum compiled quantum value. Our extension captures a wide class of Bell inequalities including tilted-CHSH, enabling us to bound the compiled value of the tilted-CHSH inequalities, by the quantum value and a negligible function of security parameter, see \cref{thm:tilt_bound} for details.

\paragraph*{A compiled self-testing result}
Our second contribution is a concept of self-testing in the compiled setting. One of the main obstacles to deriving self-testing results in the compiled setting is the lack of techniques for extracting any algebraic relations on the measurement operators acting under the encryption. Nevertheless, it remains possible to derive relations on the observables in the second round. With this in mind, we consider a partial notion of self-testing that applies to the measurements made by the prover in the second round. In particular, our definition only requires the existence of an isometry robustly certifying the ideal post-measurement state after the first round, and the action of the measurements made in the second; see \cref{def:compst} for details. 

As our final result, we provide an example by showing violations of the compiled tilted-CHSH inequalities satisfy this notion of partial self-testing. This family of inequalities was introduced by Ac\'in, Massar, and Pironio \cite{AcinRandomnessNonlocality}, and
the Bell functionals take the form $\alpha_{\theta} \langle A_{0} \rangle  + \langle A_{0}B_{0} \rangle  + \langle A_{0}B_{1} \rangle + \langle A_{1}B_{0} \rangle - \langle A_{1}B_{1} \rangle$, where $\langle A_{x}B_{y} \rangle, \ \langle A_{x} \rangle$ denote the expectation of measurements corresponding to settings $X=x, \ Y=y$, and $\alpha_{\theta} \in \mathbb{R}$. Notably, they are tailored to robustly self-test the two qubit states $\cos(\theta)\ket{00} + \sin(\theta)\ket{11}$ \cite{YangSelfTest,BampsPironio}, and were used as part of a more complex protocol to obtain self-testing for all pure bipartite entangled states~\cite{Coladangelo_2017}. The work of Barizien, Sekatski, and Bancal \cite{Barizien2024} extended this family to include extra degrees of freedom in Bob's measurements, which we will refer to as ``extended'' tilted-CHSH inequalities.

We apply \cref{thm:ExtendedSOS} to the SOS decomposition for the extended tilted-CHSH inequalities presented in~\cite{Barizien2024}. Specifically, in \cref{thm:tilt_bound} we prove that the maximum quantum value achieved for any of the extended tilted-CHSH functionals is preserved by the KLVY compilation procedure. Then in \cref{thm:tilt_st} we use this same decomposition to prove that this family of games is a compiled self-test according to \cref{def:compst}.

\paragraph*{Related work}

A recent work of \cite{kulpe24} implies that the compiled value of any 2-prover Bell scenario is bounded by the largest violation possible among so-called \emph{commuting operator} models. However, unlike some previous results, such as \cite{Baroni24,Cui24}, the upper bound in \cite{kulpe24} lacks a dependence on the security parameter $\lambda$, making it unclear how the compiled value is related to the quantum value at fixed security parameters. Hence, results such as ours, which obtain a bound on the compiled value that depends negligibly on the security parameter, remain of great importance. Furthermore, \cite{kulpe24} also considers a notion of self-testing in the compiled setting, however, due to their methods the results are in terms of commuting operator self-tests (as defined in \cite[Proposition 7.8]{paddock2024operator}) and only hold in the limit of the security parameter $\lambda \rightarrow \infty$. 

Another related work is~\cite{MV21}, which presents a protocol for certifying that an unknown computationally bounded device has prepared a maximally entangled pair of qubits, and whether a measurement was performed on each qubit in either the computational or Hadamard basis. The techniques used to prove our compiled self-test have similarities to those of~\cite{MV21}, particularly in the choice of isometry (see Definition \ref{def:compst_inf}) and proof structure, which in turn resembles self-testing techniques in the bipartite setting~\cite{BampsPironio}. There are however some key differences. Firstly,~\cite{MV21} certifies the preparation of a maximally entangled state by the device before any measurements are made. While our results are tailored to the more general class of partially entangled states, we only make statements about the post-measurement states after each round. It is an interesting open question if our results can be extended in this way (see \cref{sec:comp_self_test} for more details), and statements weaker than certifying the prepared state could also be possible. For example, can a compiled self-test be used to show the prepared state must have been entangled? Another significant difference to~\cite{MV21} is that the self-testing protocol in this work strongly resembles the bipartite case, owing to the compilation procedure mapping bipartite nonlocal scenarios to single prover scenarios. Our main result can therefore be interpreted as translating a self-testing statement in the Bell scenario to one in the compiled Bell scenario. On the other hand, the authors of~\cite{MV21} describe their approach as more ``custom'', guided by the available cryptographic primitives, and pose the open question of finding a general procedure for translating self-testing results from the nonlocal setting. We showed this is possible for the special case of titled-CHSH inequalities.

\paragraph*{Future outlook}
Moving forward, we consider several natural directions for following up on this work:

\begin{enumerate}
    \item Tilted-CHSH inequalities were an integral component of the self-testing for all pure bipartite entangled states \cite{Coladangelo_2017}. Building off of our work on compiled tilted-CHSH inequalities, a natural question is whether similar results can be obtained in the compiled setting.
    \item It would be desirable to understand the fundamental limitations of our notion of self-testing and other similar notions such as the computational self-testing given in \cite{MV21}. Furthermore, is a finer notion of self-testing in the compiled setting that characterizes both Alice's and Bob's operators and the initial state possible without specifying the underlying QHE scheme? Moreover, is every self-test in the standard Bell scenario also a compiled self-test, and vice-versa?
    \item Many current techniques for bounding the value of compiled nonlocal games/Bell inequalities can be obtained using some variant of the sum-of-squares decomposition approach. Given our improvements to this approach outlined in \cref{thm:ExtendedSOS}, it is possible to search for valid decompositions which include arbitrary words in Bob's operators. Is it possible to use this approach to give a limited variant of the NPA hierarchy \cite{NPA2} in the compiled setting? 
\end{enumerate}

\section*{Acknowledgments}
The authors thank Simon Schmidt, Ivan {\v S}upi\'c, Anand Natarajan, and Tina Zhang for helpful discussions. We also thank the anonymous TQC reviewers for valuable feedback. AM is supported by the NSERC Alliance Consortia Quantum grants program, reference number: ALLRP 578455 - 22 and the NSERC Discovery Grants program 2024-06049. CP is supported by the Digital Horizon Europe project FoQaCiA, Foundations of quantum computational advantage, grant no.~101070558, and the Natural Sciences and Engineering Research Council of Canada (NSERC) reference number: ALLRP 578455 - 22. LW is supported by the Engineering and Physical Sciences Research Council (EPSRC Grant No. EP/SO23607/1) and the European Union’s Horizon Europe research and innovation programme under the project “Quantum Security Networks Partnership” (QSNP, grant agreement No. 101114043).

\section{Background}\label{sec:comp_ineq}

\subsection{Mathematical notation}
Throughout the article, Hilbert spaces are denoted by $\mcH$, and are assumed to be finite-dimensional unless explicitly stated otherwise. Elements of $\mcH$ are denoted by $\ket{v}\in \mcH$, where the inner product $\langle u|v\rangle$ for $\ket{v},\ket{u}\in \mcH$ is linear in the second argument and defines the vector norm $\| \ket{v}\| = \sqrt{\braket{v}{v}}$. Quantum pure states are the norm 1 elements of $\mcH$. In this work, $\mathbb{B}(\mcH)$ denotes the unital $\dagger$-algebra of bounded linear operators on $\mcH$ with norm $\|M\|^2_{\mathrm{op}}=\sup_{|v\rangle\in \mcH,\ket{v} \neq 0}\langle v|M^\dagger M|v\rangle/\langle v|v\rangle$. We also write $\|A\|_2=\sqrt{\tr(A^{\dagger}A)}$ to denote the Schatten 2-norm for $A\in \mathbb{B}(\C^d)\cong M_d(\C)$. The unit in $\mathbb{B}(\mcH)$ is denoted by $\Id$, and we write $|M| = \sqrt{M^{\dagger}M}$ for the positive part of $M\in \mathbb{B}(\mcH)$. Given a finite set $\mcA$, a collection of positive operators $\{M_a\geq 0:a\in \mcA\}$ with the property that $\sum_{a\in \mcA} M_a=\Id$, is called a POVM over $\mcA$. When the operators in a POVM are orthogonal projections, we call it a PVM. Given a random variable $X$, which takes values $X=x \in \mcX$ according to a distribution $\mu:\mcX \to \R_{\geq 0}$ such that $\sum_{x\in X}\mu(x)=1$, we denote the expectation of $X$ by $\expect[X]=\sum_{x\in \mcX}\mu(x)\cdot x$. For $a,b \in \mathbb{R}$ and $\delta > 0$, $a \approx_{\delta} b$ is short for $|a - b| \leq \delta$. A function $\text{negl}:\mathbb{N} \rightarrow \mathbb{R}$ is called negligible if for all $k\in \mathbb{N}$ there exists $N \in \mathbb{N}$ such that for every $n \geq N$ it holds that $\text{negl}(n)\leq \frac{1}{n^k}$. 

\subsection{Bell scenarios, inequalities, and violations}
Before we discuss compiled Bell inequalities, let us recall the bipartite case. Here we let $\mcA,\mcB,\mcX,$ and $\mcY$ be finite sets, with $|\mcA|=m_A$, $|\mcB|=m_B$, $|\mcX|=n_A$, and $|\mcY|=n_B$. A bipartite Bell scenario is described by the tuple $\mcS=(\mcA,\mcB,\mcX,\mcY,\pi)$, where $\pi:\mcX \times \mcY\to \R_{\geq 0}$ is a distribution over the measurement settings. In a scenario, each party receives an input $x\in \mcX$ (resp. $y\in \mcY$) sampled according to $\pi$, and returns outputs $a\in \mcA$ (resp. $b\in \mcB$). The parties are non-communicating, and therefore cannot coordinate their outputs. The behaviour of the provers is characterized by a correlation, a set of conditional probabilities $\p = \{p(a,b|x,y):a\in \mcA,b\in \mcB,x\in \mcX,y\in \mcY\}$, which is realized by an underlying physical theory or model. In the quantum setting, we allow the provers to share a bipartite quantum state, and say the correlation $\p$ is realized by a \textbf{bipartite (quantum) model} 
\begin{equation}\label{eqn:q_model}
\qmodel = \big(\mcH_A,\mcH_B, \{\{M_{a|x}\}_{a \in \mathcal{A}}\}_{x \in \mathcal{X}}, \{\{N_{b|y}\}_{b \in \mathcal{B}}\}_{y \in \mathcal{Y}},\ket{\Psi}_{AB}\big),
\end{equation}
where $\mcH_{A}$ and $\mcH_{B}$ are Hilbert spaces, $\{M_{a|x}\}_{a \in \mathcal{A}}$ and $\{N_{b|y}\}_{b \in \mathcal{B}}$ are POVMs on $\mcH_A$ and $\mcH_B$ respectively, and $\ket{\Psi}_{AB}$ is a vector state in $\mcH_A \otimes \mcH_B$. More generally, a correlation $\mathsf{p}$ is quantum (or an element of $C_{\mathrm{q}}(n_A,n_B,m_A,m_B)$) if there exists a bipartite model $\qmodel$ for which $\mathsf{p}$ can be realized via the Born rule as $p(a,b|x,y) = \bra{\Psi}M_{a|x} \otimes N_{b|y} \ket{\Psi}$.
We denote the class of bipartite (quantum) models by $\mathcal{Q}(n_A,n_B,m_A,m_B)$. From now on we will refer to such models simply as \textbf{bipartite models}.

In contrast to the set of quantum correlations, we have the collection of local correlations $C_{\mathrm{loc}}(n_A,n_B,m_A,m_B)$. These are the correlations $\{p(a,b|x,y)\}$ for which there exists a \textbf{classical model}, that is a probability distribution $\mu_k$ and a local distributions $p_k^{A}(a|x)$ and $p_k^{B}(b|y)$ such that $p(a,b|x,y) = \sum_{k} \mu_k \, p_k^{A}(a|x) \, p_k^{B}(b | y)$.
We let $\cmodel=(\mu_k,\{p_k^A\},\{p_k^B\})$ denote a \textbf{classical model} and let $\mathcal{L}(n_A,n_B,m_A,m_B)$ denote the class of all classical models. In what follows we consider Bell scenarios where $n_A=n_B=n$, and $m_A=m_B=m$. With this notation Bell's theorem~\cite{Bell64} states that $C_{\text{loc}}(2,2)$ is a strict subset of $C_{\text{q}}(2,2)$.

Given a Bell scenario $\mcS$, one can consider a linear (or Bell) functional on the set of correlations
\begin{equation}\label{eqn:bell_ineq}
    I=\sum_{a\in \mathcal{A},b\in \mathcal{B},x\in \mathcal{X},y\in \mathcal{Y}} w_{abxy} \, p(a,b|x,y),   
\end{equation}
for coefficients $w_{abxy} \in \mathbb{R}$. A \textbf{Bell inequality} is a functional $I$ and a bound $\eta>0$ such that $I\leq \eta$ for all $\p \in C_{\mathrm{loc}}(n,m)$. Given a functional $I$, the classical value is the maximal value achieved by the classical correlations $\p \in C_{\mathrm{loc}}(n,m)$. We denote this value by $\eta^{\mathrm{L}}:= \sup_{\p\in C_{\mathrm{loc}}(m,n)}I$. The quantum value for $I$ is the maximal value achieved by the set of quantum correlations $\p \in C_{\mathrm{q}}(m,n)$, and we denote the quantum value on $I$ by $\eta^{\mathrm{Q}}:= \sup_{\p\in C_{\mathrm{q}}(m,n)}I$. Hence, a Bell violation occurs whenever there is a $\p\in C_{\mathrm{q}}(m,n)$ for which $I>\eta^{\mathrm{L}}$. A violation of a Bell inequality by non-communicating provers employing a quantum model is an indication of entanglement between provers.

Typically when $\eta^{\mathrm{L}}$ is known for a given $I$, the main challenge is finding an upper bound on $\eta^{\mathrm{Q}}$. In this case, one often considers the \textbf{Bell operator}\footnote{For a more mathematically rigorous treatment of Bell operators and the SOS approach consult \cite{Cui24}.} $S = \sum_{abxy}w_{abxy}\, M_{a|x} \otimes N_{b|y}$, and $\langle S \rangle = \bra{\Psi} S \ket{\Psi}$ its quantum expectation with respect to $|\Psi\rangle \in \mcH_A\otimes \mcH_B$. Since bipartite models with separable quantum states generate the classical correlations $C_{\mathrm{q}}(m,n)$, $\langle \Psi| S|\Psi\rangle\leq \eta^{\mathrm{L}}$ whenever $|\Psi\rangle$ is separable (unentangled). However, it's possible that there could be entangled states for which $\langle \Psi'|S|\Psi'\rangle>\eta^{\mathrm{L}}$. Hence, given a Bell operator $S$, we can recover the maximum classical and quantum values $\eta^{\mathrm{L}} = \sup_{\cmodel \in \mcL(n,m)} \langle S \rangle$ and $\eta^{\mathrm{Q}} = \sup_{\qmodel \in \mathcal{Q}(n,m)} \langle S \rangle$ respectively. Technically, we have not fixed the dimensions of the Bell operator as we want to consider any finite-dimensional model. Hence, the supremum is implicitly over all finite-dimensional Hilbert spaces $\mcH_A\otimes \mcH_B$.

An approach to establishing upper bounds on $\langle S \rangle$ is using sum-of-squares techniques. Let $S$ be a Bell operator and $\eta'>0$. The shifted Bell operator $\eta'\Id - S$ admits a \textbf{sum-of-squares} (SOS) decomposition if there exists a set of polynomials $\{P_{i}\}_{i\in \mcI}$ in the elements $\{M_{a|x}, N_{b|y}:a\in \mcA, b\in \mcB, x\in \mcX, y\in \mcY\}$ satisfying $\eta'\Id - S = \sum_{i\in \mcI}P_{i}^{\dagger}P_{i}$.
The existence of an SOS decomposition for the operator $\eta' \Id - S$ implies that $\eta' \Id - S$ is positive, and therefore $\eta'$ is an upper bound on the maximum quantum value of $\langle S \rangle$. Additionally, if $\eta'$ is achievable by a bipartite model, then we write $\eta' = \eta^{\mathrm{Q}}$. In this case, the \emph{shifted} Bell operator is $\bar{S} = \eta^{\mathrm{Q}} \Id - S$, and observing $\bra{\Psi} \bar{S} \ket{\Psi} = 0 $ implies the constraints $P_{i} \ket{\Psi} = 0$ for all $i \in \mcI$; these constraints can often be used to infer the algebraic structure (rigidity) of the measurements $\{M_{a|x}\}_{a \in \mcA,x \in \mcX},\{N_{b|y}\}_{b \in \mcB,y \in \mcY}$ which achieve $\langle S \rangle = \eta^{\mathrm{Q}}$.

\section{Compiled Bell scenarios}
\label{sec:comp}
The compilation procedure of a Bell scenario is essentially the same as the procedure for compiling nonlocal games outlined in \cite{KLVY23}. Let $\mcS=(\mcX,\mcY,\mcA,\mcB,\pi)$ be a 2-prover Bell scenario and fix a quantum homomorphic encryption scheme with \emph{security against quantum distinguishers} and \emph{correctness with respect to auxiliary input}. Readers unfamiliar with QHE schemes and these properties can refer to \cref{def_qhe} found in the appendix.

A \textbf{compiled Bell scenario} is the following 2-round single-prover scenario. To setup, the verifier samples a secret key $\sk\leftarrow \gen(1^\lambda)$. Then, the verifier samples a pair of inputs $(x,y)\in \mcX\times \mcY$ according to the distribution $\pi:\mcX\times \mcY\to \R_{\geq 0}$, and encrypts the first input as the ciphertext $\chi\leftarrow \enc(\sk,x)$.
\begin{enumerate}
    \item The verifier sends the ciphertext $\chi$ to the prover. The prover replies with a ciphertext $\alpha$ encoding their output. The verifier decrypts obtaining outcome $a\leftarrow \dec(\sk,\alpha)$ from $\mcA$.
    \item The verifier sends the sampled (plaintext) input $y\in \mcY$ to the prover, who replies with another outcome $b\in \mcB$.
\end{enumerate}

In the compiled scenario, for a chosen security parameter $\lambda$, the prover prepares an initial quantum polynomial time (QPT) preparable state $\ket{\Psi^{(\lambda)}} \in \widetilde{\mcH}^{(\lambda)}$ where $\widetilde{\mcH}^{(\lambda)}$ is a single Hilbert space (see \cref{def:QPT} for details on efficient quantum procedures). Then, the first round of the protocol is characterized by a family of POVMs $\{\{\widetilde{M}_{\alpha | \chi}^{(\lambda)}\}_{\alpha \in \bar{\mathcal{A}}}\}_{\chi \in \bar{\mathcal{X}}}$ and unitaries $\{U_{\alpha,\chi}^{(\lambda)}\}_{\alpha \in \bar{\mathcal{A}},\chi \in \bar{\mathcal{X}}}$, where $\bar{\mathcal{X}}$ and $\bar{\mathcal{A}}$ are the set of all valid ciphertexts of the first round input and output, respectively. Unlike in the bipartite setting, we must account for unitary operations applied to the post-measurement state in the first round. With this in mind, we denote the sub-normalized post-measurement state given the measurement over ciphertext $\chi$ and encrypted outcome $\alpha$ by
\begin{equation}
 U_{\alpha,\chi}^{(\lambda)}\widetilde{M}_{\alpha|\chi}^{(\lambda)}\ket{\Psi^{(\lambda)}} =: \ket{\Psi_{\alpha | \chi}^{(\lambda)}}.
\end{equation}
Note that these vectors are sub-normalized. In particular, the probability of obtaining $\alpha \in  \bar{\mcA}$ given $\chi \in \bar{\mcX}$ is given by $\braket{\Psi_{\alpha | \chi}^{(\lambda)}}{\Psi_{\alpha | \chi}^{(\lambda)}}$.
In the second round, the device makes a POVM measurement $\{\{N_{b|y}^{(\lambda)}\}_{b \in \mathcal{B}}\}_{y \in \mathcal{Y}}$, where the resulting conditional probability is given by
\begin{equation}
\bra{\Psi^{(\lambda)}}{\widetilde{M}_{\alpha|\chi}^{(\lambda)\dagger}} {U_{\alpha,\chi}^{(\lambda)\dagger}} N_{b|y}^{(\lambda)}U_{\alpha,\chi}^{(\lambda)}\widetilde{M}_{\alpha|\chi}^{(\lambda)}\ket{\Psi^{(\lambda)}}=\bra{\Psi_{\alpha|\chi}^{(\lambda)}}N_{b|y}^{(\lambda)}\ket{\Psi_{\alpha|\chi}^{(\lambda)}},
\end{equation}
for a fixed, $\lambda\in \N$, $\sk\leftarrow \gen(1^\lambda)$, ciphertexts $\chi\in \bar{\mcX}$, $\alpha \in \bar{\mcA}$, and plaintexts $y\in \mcY$, $b\in \mcB$. 

To summarize, for a fixed QHE scheme, $\lambda\in \N$, a \textbf{compiled (quantum) model} is given by a tuple 
\begin{equation}
    \widetilde{\qmodel}^{(\lambda)} =( \widetilde{\mcH}^{(\lambda)}, \{\ket{\Psi_{\alpha|\chi}^{(\lambda)}}\}_{\alpha \in \bar{\mcA},\chi \in \bar{\mcX}}, \{\{N_{b|y}^{(\lambda)}\}_{b \in \mcB}\}_{y \in \mcY}),
\end{equation}
where all the relevant measurements and states are obtained by some QPT procedure. We remark that one can consider a description of the model which includes the initial state $\ket{\Psi^{(\lambda)}}$ and the operators $\{U_{\alpha,\chi}^{(\lambda)}\widetilde{M}_{\alpha|\chi}^{(\lambda)}\}_{\alpha \in \bar{\mcA},\chi \in \bar{\mcX}}$, rather than the post-measurement states $\ket{\Psi_{\alpha|\chi}^{(\lambda)}}$. Hence, $\widetilde{\qmodel}^{(\lambda)}$ is really a coarse description of a quantum model in the compiled setting.
The joint distribution of the outcomes after both rounds is given by
\begin{equation}
    p^{(\lambda)}(a,b|x,y) = \expect_{\sk\leftarrow \gen(1^\lambda)} \expect_{\chi : \enc(x) = \chi} \sum_{\alpha : \dec(\alpha) = a} \bra{\Psi_{\alpha|\chi}^{(\lambda)}}N_{b|y}^{(\lambda)}\ket{\Psi_{\alpha|\chi}^{(\lambda)}}. \label{eq:cryptDist}
\end{equation}

Note that the marginal distribution $p^{(\lambda)}(a|x)$ obtained from \cref{eq:cryptDist} will be independent of the second input $y$ due to the sequential nature of the protocol. However, the marginal $p^{(\lambda)}(b|y,x)$ currently depends on $x$. The aim of what follows is to establish a computational independence between this distribution and the inputs $x$. To do so we will need to consider the distributions of the decrypted outputs and appeal to the security promise of the QHE scheme. Specifically, we require a key lemma which has appeared in several works \cite{NZ23,Cui24,Baroni24}. We borrow a version from \cite{kulpe24} and we refer the reader to the reference for the proof.

\begin{lemma}[\cite{kulpe24}, Proposition 4.6]
Let $\widetilde{\mathcal{\qmodel}}^{(\lambda)}$ be a compiled quantum model, and $\mathcal{N}^{(\lambda)} = w(\{N_{b|y}^{(\lambda)}\}_{b \in \mathcal{B},y \in \mcY})$ be a monomial in the measurement operators $\{N_{b|y}^{(\lambda)}\}_{b \in \mathcal{B},y \in \mcY}$, where $\lambda \in \N$ is the security parameter for a fixed QHE scheme. Then, for any two QPT sampleable distributions $\mcD_1,\mcD_2$ over plaintext inputs $x \in \mcX$ there exists a negligible function $\negl(\lambda)$ of the security parameter $\lambda$ such that the following holds
\begin{align*}
    &\left|\expect_{\sk\leftarrow \gen(1^\lambda)} \expect_{x\leftarrow \mcD_1}\expect_{\substack{\chi:\enc(x)=\chi}}\sum_{\alpha \in \bar{\mcA}}\langle \Psi_{\alpha|\chi}^{(\lambda)}| \mcN^{(\lambda)} |\Psi_{\alpha|\chi}^{(\lambda)}\rangle - \expect_{\sk\leftarrow \gen(1^\lambda)} \expect_{x\leftarrow \mcD_2}\expect_{\substack{\chi:\enc(x)=\chi}}\sum_{\alpha \in \bar{\mcA}}\langle\Psi_{\alpha|\chi}^{(\lambda)}|\mcN^{(\lambda)}|\Psi_{\alpha|\chi}^{(\lambda)}\rangle  \right|\\
    &\leq \negl(\lambda).
\end{align*} \label{lem:QHE1}
\end{lemma}
\noindent The approximate no-signalling conditions from Alice to Bob can then be seen by applying \cref{lem:QHE1} to the monomials of degree 1 in the QPT measurement operators $\{N_{b|y}^{(\lambda)}\}_{b \in \mcB,y \in \mcY}$, since
\begin{equation}
    \Big| \expect_{\sk\leftarrow \gen(1^\lambda)}\expect_{\chi : \enc(x) = \chi} \sum_{\alpha \in \bar{\mathcal{A}}} \bra{\Psi_{\alpha|\chi}^{(\lambda)}}N_{b|y}^{(\lambda)} \ket{\Psi_{\alpha|\chi}^{(\lambda)}} - \expect_{\sk\leftarrow \gen(1^\lambda)}\expect_{\chi : \enc(x') = \chi} \sum_{\alpha \in \bar{\mathcal{A}}} \bra{\Psi_{\alpha|\chi}^{(\lambda)}}N_{b|y}^{(\lambda)} \ket{\Psi_{\alpha|\chi}^{(\lambda)}} \Big| \leq \negl(\lambda)
\end{equation}
holds for all $b \in \mathcal{B},y \in \mathcal{Y}$ and $x,x' \in \mathcal{X}$ with $x \neq x'$. 

In the above statements, the measurements are completely general, and the states are sub-normalized vectors. The following lemma shows that when considering the compiled value, we can assume that the states and measurement operators in the compiled strategy are pure and projective. 

\begin{lemma}\label{lem:proj_comp}
    Let $\mathcal{H}^{'(\lambda)}$ be the Hilbert space of the device, and $\{\{\rho_{\alpha|\chi}^{(\lambda)}\}_{\alpha \in \bar{\mcA}}\}_{\chi \in \bar{\mcX}}$ be a family of QPT-preparable sub-normalized states on $\mcH^{'(\lambda)}$ after the first round. Let $\{\{N_{b|y}^{'(\lambda)}\}_{b \in \mcB}\}_{y \in \mcY}$ be a family of QPT-implementable POVMs on $\mathcal{H}^{'(\lambda)}$, which induce the behaviour $p^{(\lambda)}(\alpha,b|\chi,y) = \tr[N_{b|y}^{'(\lambda)}\rho_{\alpha|\chi}^{(\lambda)}]$. Then there exists a Hilbert space $\mcH^{(\lambda)}$, a family of QPT-preparable sub-normalized states $\{\{\ket{\Psi_{\alpha|\chi}^{(\lambda)}}\}_{\alpha \in \bar{\mcA}}\}_{\chi \in \bar{\mcX}}$ in $\mcH^{(\lambda)}$, and a family of QPT-implementable PVMs $\{\{N_{b|y}^{(\lambda)}\}_{b \in \mcB}\}_{y \in \mcY}$ on $\mcH^{(\lambda)}$ which satisfy
    \begin{equation}
        \bra{\Psi_{\alpha|\chi}^{(\lambda)}}N_{b|y}^{(\lambda)}\ket{\Psi_{\alpha|\chi}^{(\lambda)}} = p^{(\lambda)}(\alpha,b|\chi,y), \  \ \forall \alpha \in \bar{\mcA},\, \chi \in \bar{\mcX}, \, b \in \mcB, \, y \in \mcY.
    \end{equation} \label{lem:proj}
\end{lemma}
\noindent See \cref{sec:app_b} for the proof of \cref{lem:proj_comp}. 

We say a compiled model $\widetilde{\mathcal{\qmodel}}^{(\lambda)}= ( \widetilde{\mcH}^{(\lambda)}, \{\ket{\Psi_{\alpha|\chi}^{(\lambda)}}\}_{\alpha \in \bar{\mcA},\chi \in \bar{\mcX}}, \{\{N_{b|y}^{(\lambda)}\}_{b \in \mcB}\}_{y \in \mcY})$ is pure and projective whenever the states $|\Psi_{\alpha|\chi}^{(\lambda)}\rangle$ are all pure and the measurements $N_{b|y}^{(\lambda)}$ are all projective (i.e.~PVMS).

\subsection{Quantum bounds for compiled inequalities}

A compiled (quantum) model $\widetilde{\mathcal{\qmodel}}^{(\lambda)}$ describes the correlations $\mathsf{p}^{(\lambda)} = \{p^{(\lambda)}(a,b|x,y)\}_{a \in \mcA,b \in \mcB,x \in \mcX,y \in \mcY}$ observed in a compiled Bell scenario. A \textbf{compiled Bell functional} is a linear functional $I^{(\lambda)}$ evaluated on correlations realized by compiled models. That is
\begin{equation}\label{eq:comp_bell}
I^{(\lambda)}= \sum_{abxy}w_{abxy}\expect_{\substack{\sk\leftarrow \gen(1^\lambda)\\ \chi : \enc(x) = \chi}} \sum_{\alpha : \dec(\alpha) = a} \bra{\Psi^{(\lambda)}_{\alpha|\chi}}N^{(\lambda)}_{b|y}\ket{\Psi^{(\lambda)}_{\alpha|\chi}}.
\end{equation}
By the properties of the compilation procedure~\cite[Theorem 3.2]{KLVY23}, Bell inequalities are preserved under compilation (up to negligible error). In particular, for large security parameter, efficient classical provers cannot violate a Bell inequality by much more than they could in the (bipartite) scenario.
From now on, we will suppress the security parameter $\lambda\in \N$ along with the expectation over secret keys $\expect_{\sk\leftarrow \gen(1^\lambda)}$ and simply write the expectation for a fixed key. In particular, we express the compiled model as $\widetilde{\mathcal{\qmodel}}$ and  \cref{eq:comp_bell} as
\begin{equation*}
I=\sum_{abxy}w_{abxy}\expect_{\chi : \enc(x) = \chi} \sum_{\alpha : \dec(\alpha) = a} \bra{\Psi_{\alpha|\chi}}N_{b|y}\ket{\Psi_{\alpha|\chi}}.
\end{equation*}

We now turn our attention to the maximum value $I$ can take in the compiled setting with an efficient quantum prover. The results of \cite{KLVY23} imply that an efficient quantum prover can achieve the same violation in the bipartite setting. However, the existence of a quantum compiled behavior which exceeds the maximal quantum Bell violation in the bipartite case (by more than negligible factors) has not been ruled out. Nonetheless, in several cases (like the CHSH inequality and more generally all XOR games \cite{Cui24}) we know that the quantum compiled behavior cannot exceed the value $\eta^{\mathrm{Q}}$ by more than negligible amounts. One technique for establishing such bounds was introduced in \cite{NZ23} and uses SOS techniques to bound the quantum violation of the compiled Bell functional.

\subsection{Extending the pseudo-expectations}

Our approach builds off the methods used in \cite{NZ23} and \cite{Cui24}. To explain this approach we recall that a pseudo-expectation is a unital, linear map from a subspace $\mathcal{T}$ of the algebra generated by $\{M_{a|x}, N_{b|y}\}_{a\in \mathcal{A},x \in \mathcal{X},b \in \mathcal{B},y \in \mathcal{Y}}$ to the complex numbers, $\pseudo:\mathcal{T} \to \mathbb{C}$, which is determined by a compiled quantum model $\widetilde{\mathcal{\qmodel}}$. In the case $n=m=2$, it suffices to define the pseudo-expectation $\pseudo$ on the observables $A_{x} = \sum_{a \in \{0,1\}} (-1)^{a}M_{a|x}, \ B_{y} = \sum_{b \in \{0,1\}}(-1)^{b}N_{b|y}$ and require that they are mapped to their expectations in the compiled scenario\footnote{Though in the following we define $\pseudo$ for $n=m=2$, this can be directly extended to arbitrary Bell scenarios by defining $\pseudo$ on the POVM elements $M_{a|x}, N_{b|y}$ in an analogous way.}. We further assume that all measurements are projective (cf. \cref{lem:proj}). In previous works, the definition of the pseudo-expectation had been restricted to monomials consisting of at most one Alice and one Bob observable as outlined below:
\begin{equation}
    \begin{aligned}
        \pseudo[A_{x}B_{y}] &:= \expect_{\chi : \enc(x) = \chi} \sum_{\alpha} (-1)^{\dec(\alpha)}\bra{\Psi_{\alpha|\chi}}B_{y} \ket{\Psi_{\alpha|\chi}}, \\
        \pseudo[A_{x}A_{x'}] &:= \delta_{x,x'}, \\
        \pseudo[B_{y}B_{y'}] &:= \expect_{x \in \mcX} \expect_{\chi : \enc(x) = \chi} \sum_{\alpha} \bra{\Psi_{\alpha|\chi}}B_{y}B_{y'} \ket{\Psi_{\alpha|\chi}}, \\
        \pseudo[A_{x}] &:= \expect_{\chi : \enc(x) = \chi} \sum_{\alpha} (-1)^{\dec(\alpha)} \braket{\Psi_{\alpha|\chi}}{\Psi_{\alpha|\chi}}, \\
        \pseudo[B_{y}] &:= \expect_{x \in \mcX} \expect_{\chi : \enc(x) = \chi} \sum_{\alpha} \bra{\Psi_{\alpha|\chi}}B_{y} \ket{\Psi_{\alpha|\chi}}, \\
        \pseudo[\Id] &:= 1,
    \end{aligned}
\end{equation} 
where $\expect_{x \in \mcX}$ denotes the expectation according to an arbitrary fixed distribution over $\mathcal{X}$. This is already sufficient to handle known SOS decompositions for a variety of well-studied Bell inequalities whenever the polynomials are expressed in the basis $\{\Id,A_{x},B_{y}\}_{x\in \mcX,y \in \mcY}$. However, there are Bell inequalities, such as the tilted-CHSH inequality~\cite{BampsPironio,Barizien2024}, for which no known SOS decomposition exists in the basis $\{\Id,A_{x},B_{y}\}_{x \in \mcX,y \in \mcY}$. 

The contribution of this section is to expand the definition of the pseudo-expectation to the basis encompassing all monomials in $A_{x},B_{0},B_{1}$, for a fixed $x \in \mathcal{X}$, in a way that is approximately non-negative on Hermitian squares. This allows us to handle more general SOS decompositions, and in particular, the tilted-CHSH inequalities. Let $ w(A_{x},B_{0},B_{1})$ be a monomial in the elements $\{A_{x},B_{0},B_{1}\}$. Importantly, $x$ is fixed, and we do not consider monomials of the form $A_{0}A_{1}B_{y}$ for example. Let $\bar{w}$ be the canonical form of $w$ under the relations $[A_{x},B_{y}] = 0$, $(B_{y})^{2} = (A_{x})^{2} = \Id$, where all $A_{x}$ terms are commuted to the left. Since we only consider one value of $x$, these will all be of the form $ (A_{x})^{i} \bar{w}(B_{0},B_{1})$ for some $i \in \{0,1\}$, where the monomial $\bar{w}(B_{0},B_{1})$ cannot be reduced further. We then define the pseudo-expectation
\begin{equation}
    \pseudo\big[ w(A_{x},B_{0},B_{1}) \big] := \pseudo\big[ (A_{x})^{i} \bar{w}(B_{0},B_{1}) \big].
\end{equation}
For the case $i = 0$, we define
\begin{equation}
    \pseudo\big[ \bar{w}(B_{0},B_{1}) \big] := \expect_{x \in \mathcal{X}}\expect_{\chi:\enc(x)=\chi} \sum_{\alpha} \bra{\Psi_{\alpha|\chi}}\bar{w}(B_{0},B_{1}) \ket{\Psi_{\alpha|\chi}}, \label{eq:pseduo1}
\end{equation}
and for the case $i = 1$, 
\begin{equation}
    \pseudo\big[ A_{x}\bar{w}(B_{0},B_{1}) \big] := \expect_{\chi:\enc(x)=\chi} \sum_{\alpha}(-1)^{\dec(\alpha)} \bra{\Psi_{\alpha|\chi}}\bar{w}(B_{0},B_{1}) \ket{\Psi_{\alpha|\chi}}. \label{eq:pseduo2}
\end{equation}
From the above definitions, we next state the main result of this section, which can be applied generally to any polynomial expressible in the basis $\{A_{x},B_{0},B_{1}\}$.
\begin{theorem}\label{thm:ExtendedSOS}
    Let $\{A_{x}\}_{x \in \mcX}$ and $\{B_{y}\}_{y \in \mathcal{Y}}$ be binary observables, and let 
    \begin{equation}
        P = \sum_{i}\gamma_{i}(A_{x})^{k_{i}}w_{i}(B_{0},B_{1}),
    \end{equation}
    where $\gamma_{i} \in \mathbb{C}$, $k_{i} \in \{0,1\}$ and each $w_{i}(B_{0},B_{1})$ is any monomial in the algebra of $\{B_{0},B_{1}\}$. Then there exists a negligible function $\negl(\lambda)$ of the security parameter $\lambda \in \mathbb{N}$ such that 
    \begin{equation}
        \pseudo[P^{\dagger}P] \geq - \negl(\lambda).
    \end{equation} \label{lem:compSOS} Furthermore, for a given Bell functional $I$, and a compiled model $\widetilde{\mathcal{\qmodel}}$,  $\pseudo(I)$ is the expected value of the compiled model $\widetilde{\mathcal{\qmodel}}$ on $I$.
\end{theorem}

The proof can be found in \cref{sec:app_b}.

\subsection{Quantum bounds for compiled tilted-CHSH expressions}

We now present the family of extended tilted-CHSH type expressions and their SOS decompositions discovered in~\cite{Barizien2024}. Let $\theta \in (0,\pi/4]$, $\phi \in \big( \max\{-2\theta,-\pi + 2\theta\} , \min\{2\theta,\pi - 2\theta\}\big) \setminus \{0\}$, and $\tau_{\theta,\phi} \in \mathbb{R}$ such that
\begin{equation}
    \frac{1}{\tau_{\theta,\phi}^{2}}  = \frac{\sin^{2}(2\theta)}{\tan^{2}(\phi)} - \cos^{2}(2\theta). \label{eq:lam}  
\end{equation}
From here, we define the following expressions:
\begin{equation}
\begin{aligned}
    S_{\theta,\phi} &:= A_{0}\otimes \frac{B_{0}+B_{1}}{\cos(\phi)} + \tau_{\theta,\phi}^{2} \Big [ \sin(2\theta) \, A_{1}\otimes \frac{B_{0} - B_{1}}{\sin(\phi)} + \cos(2\theta) \, \Id \otimes \frac{B_{0} + B_{1}}{\cos(\phi)}  \Big], \\ \eta^{\mathrm{Q}}_{\theta,\phi} &:= 2(1+\tau_{\theta,\phi}^{2}).\label{eq:tiltBop}
\end{aligned}
\end{equation}
We also let $I_{\theta,\phi}$ denote the corresponding Bell functional, and recall the following result.
\begin{lemma}[\cite{Barizien2024}, Section 3.2.1]
    Let $\theta \in (0,\pi/4]$, $\phi \in \big( \max\{-2\theta,-\pi + 2\theta\} , \min\{2\theta,\pi - 2\theta\}\big)  \setminus \{0\} $, $\tau_{\theta,\phi}$ be given by \cref{eq:lam} and $S_{\theta,\phi}, \eta^{\mathrm{Q}}_{\theta,\phi}$ be defined in \cref{eq:tiltBop}. Define the following polynomials:
    \begin{equation}
    \begin{aligned}
        N_{0} &:= A_{0}\otimes \Id - \Id \otimes \frac{B_{0} + B_{1}}{2\cos(\phi)}, \\
        N_{1} &:= A_{1} \otimes \Id - \sin(2\theta) \, \Id \otimes \frac{B_{0} - B_{1}}{2\sin(\phi)} - \cos(2\theta) \, A_{1} \otimes \frac{B_{0} + B_{1}}{2\cos(\phi)}.
    \end{aligned} \label{eq:SOSpoly}
    \end{equation}
    Then the shifted Bell operator $\bar{S}_{\theta,\phi} = \eta^{\mathrm{Q}}_{\theta,\phi}\Id - S_{\theta,\phi}$ admits the SOS decomposition 
    \begin{equation}
        \bar{S}_{\theta,\phi} = N_{0}^{\dagger}N_{0} + \tau_{\theta,\phi}^{2}N_{1}^{\dagger}N_{1}.
    \end{equation} \label{lem:tilt}
\end{lemma}
Using the decomposition in \cref{lem:tilt}, it was shown in~\cite{Barizien2024} that the inequality $\langle S_{\theta,\phi} \rangle \leq \eta_{\theta,\phi}^{\mathrm{Q}}$ self-tests following state and measurements:
\begin{equation}
\begin{aligned}
\ket{\psi_{\theta}} &= \cos(\theta)\ket{00} + \sin(\theta)\ket{11},\\
    A_{0} &= \sigma_{Z}, \ A_{1} = \sigma_{X}, \\
    B_{y} &= \cos(\phi) \, \sigma_{Z} + (-1)^{y} \sin(\phi) \, \sigma_{X}, \ y \in \{0,1\}, \label{eq:honCHSH}
\end{aligned}
\end{equation}
where $\sigma_{Z},\sigma_{X}$ are the Pauli operators. Notably, by setting $\phi = \mu_{\theta}$ where $\tan(\mu_{\theta}) = \sin(2\theta)$, this family encompasses what are most commonly referred to as ``tilted-CHSH inequalities'' given by the Bell operator
\begin{equation}
    T_{\theta} = \alpha_{\theta} A_{0} \otimes \Id + A_{0}\otimes (B_{0} + B_{1}) + A_{1}\otimes (B_{0} - B_{1}), \label{eq:tilt}
\end{equation}
where $\alpha_{\theta} = 2/\sqrt{1+2\tan^{2}(2\theta)}$ ~\cite{AcinRandomnessNonlocality,YangSelfTest,BampsPironio}. Compared to the SOS decompositions for $T_{\theta}$ from~\cite{BampsPironio}, the decomposition of~\cite{Barizien2024} is expressed in the basis for which our extended pseudo-expectation is well defined (cf. \cref{lem:compSOS}), allowing us to provide bounds on the compiled value of $T_{\theta}$, and more generally the family $S_{\theta,\phi}$.

\begin{theorem}\label{thm:tilt_bound}
Let $\theta \in (0,\pi/4]$, $\phi \in \big( \max\{-2\theta,-\pi + 2\theta\} , \min\{2\theta,\pi - 2\theta\}\big)  \setminus \{0\} $, and let $\langle S_{\theta,\phi} \rangle \leq \eta_{\theta,\phi}^{\mathrm{Q}}$ be the extended tilted-CHSH inequality for the expressions defined in \cref{eq:tiltBop}. Then the maximum quantum value of the corresponding compiled Bell inequality is given by $\eta^{\mathrm{Q}}_{\theta,\phi} + \negl(\lambda)'$, where $\negl(\lambda)'$ is a negligible function of the security parameter.
\end{theorem}

\begin{proof}
    We evaluate the pseudo-expectation on the shifted Bell expression $\bar{S}_{\theta,\phi}$:
    \begin{equation}
        \begin{aligned}
            \pseudo[\bar{S}_{\theta,\phi}] &= \pseudo[N_{0}^{\dagger}N_{0}] + \tau^{2}_{\theta,\phi}    \pseudo[N_{1}^{\dagger}N_{1}],
        \end{aligned}
    \end{equation}
    where we used the decomposition in \cref{lem:tilt}. The polynomial $N_{0}$ is expressed in the basis $\{A_{0},B_{0},B_{1}\}$, and we find by \cref{lem:compSOS} that 
    \begin{equation}
        \pseudo[N_{0}^{\dagger}N_{0}] \geq - \negl(\lambda).
    \end{equation}
    Similarly, $N_{1}$ is expressed in the basis $\{A_{1},B_{0},B_{1}\}$, and we see by \cref{lem:compSOS} that $\pseudo[N_{1}^{\dagger}N_{1}] \geq - \negl(\lambda)$. Putting these together, we obtain 
    \begin{equation}
        \pseudo[\bar{S}_{\theta,\phi}] \geq -\negl(\lambda)(1+\tau^{2}_{\theta,\phi}) =: -\negl(\lambda)',
    \end{equation}
    which implies $\pseudo[S_{\theta,\phi}] \leq \eta^{\mathrm{Q}}_{\theta,\phi} + \negl(\lambda)'$ as desired, where $\pseudo[S_{\theta,\phi}]$ is the expected value of the compiled Bell inequality. 
\end{proof}

\begin{remark}
    The extension of the $S_{\theta,\phi}$ family presented in \cite[Section 3.2.3]{Barizien2024} self-tests the state $\ket{\psi_{\theta}}$ along with the more general measurements
\begin{equation}
\begin{aligned} 
    A_{0} &= \sigma_{Z}, \ A_{1} = \sigma_{X}, \\
    B_{0} &= \cos(\phi) \, \sigma_{Z} +  \sin(\phi) \, \sigma_{X}, \\
    B_{1} &= \cos(\omega) \, \sigma_{Z} +  \sin(\omega) \, \sigma_{X},
\end{aligned}
\end{equation}
for $\phi \in (-2\theta,0)$ and $\omega \in (0,2\theta)$. This family of Bell inequalities can also be compiled under our definition of the pseudo-expectation. This is because each SOS polynomial is given in the basis $\{A_{x},B_{0},B_{1}\}$ for a fixed $x$, and we can apply \cref{lem:compSOS} directly as was done in \cref{thm:tilt_bound}. We omit the explicit proof of this for brevity.      
\end{remark}

\section{Self-testing in the compiled setting}\label{sec:selftest}

Recall that a bipartite (quantum) model $\mathrm{Q}$, consists of a shared state $\ket{\Psi}$, along with local POVM measurements $\{M_{a|x}\}$ and $\{N_{b|y}\}$ for Alice and Bob, respectively. Given a Bell expression $I$, the inequality $I \leq \eta^{Q}$ self-tests an ideal bipartite model $\mathrm{Q}^*$ if any optimal bipartite model is essentially the same as $\mathrm{Q}^*$, modulo some physically irrelevant degrees of freedom. This is more formally stated in terms or the existence of local isometries which maps the employed model to the ideal one. When small errors are permitted, one considers the following definition of robust self-testing.

\begin{definition}[Bipartite self-test]
    The inequality $I \leq \eta^{\mathrm{Q}}$ is a self-test for a bipartite model $\mathrm{Q}^* = \left( \lbrace P_{a|x} \rbrace , \lbrace Q_{b|y} \rbrace, \ket{\phi} \right)$ if there exist a non-negative function $f(\epsilon)$ such that $f(\epsilon) \to 0$ as $\epsilon \to 0$, such that for any bipartite model $\mathrm{Q} =  \left( \lbrace M_{a|x} \rbrace , \lbrace N_{b|y} \rbrace, \ket{\Psi} \right) $ achieving $I \geq \eta^{\mathrm{Q}} - \epsilon$ for $\epsilon \geq 0$, there exists a Hilbert space $\mathcal{H}_{\mathrm{aux}}$, an auxiliary state $\ket{\zeta} \in \mathcal{H}_{\mathrm{aux}}$ and local isometries $V_A$ and $V_B$, such that defining $V:\mathcal{H}_{A} \otimes \mathcal{H}_{B} \to \mathbb{C}^{d} \otimes \mathbb{C}^{d} \otimes \mathcal{H}_{\mathrm{aux}}$, $V = V_{A} \otimes V_{B}$, the following is satisfied for all $x,y,a,b$:
\begin{equation*}
        \big \| V_A \otimes V_B (M_{a|x} \otimes N_{b |y} ) \ket{\Psi} - (P_{a|x} \otimes Q_{b|y})\ket{\phi} \otimes \ket{\zeta} \big \| \leq f(\epsilon).
\end{equation*}
\end{definition}

In the bipartite setting, one could consider the situation where Alice measures first using a POVM $\{P_{a|x}\}$, collapsing the state to a post-measurement state $\rho_{a|x}$ on Bob's subsystem $\mcH_B$, upon which Bob performs his measurement, resulting in the application of the POVM element $Q_{b|y}$. With this in mind, we consider the setting where the only relevant features of the model are those from Bob's (resp. Alice's) perspective. In particular, subsystem $A$ is traced out following the recorded measurement of outcome of $a$ given $x$.

\begin{definition}[Partial model]
    Given a bipartite model $\mathrm{Q} = \left( \lbrace M_{a|x} \rbrace , \lbrace N_{b|y} \rbrace, \ket{\Psi} \right)$, we define the partial model of $\mathrm{Q}$ by $\mathrm{Q}' = (\{N_{b|y}\}, \{\rho_{a|x}\})$ where
    \begin{equation}
        \rho_{a|x} = \tr_{A}[(M_{a|x} \otimes \Id_{B})\ketbra{\Psi}{\Psi}].
    \end{equation}
    We note that $\rho_{a|x}$ will generally be mixed. When each $\rho_{a|x}$ is pure, we say that $\mathrm{Q}$ has a pure partial model, denoted by $\mathrm{Q}' = (\{N_{b|y}\}, \{\ket{\phi_{a|x}}\})$.
\end{definition}

Symmetrically, given a bipartite model one can consider a (pure) partial model on $\mcH_A$ by tracing out subsystem $B$. However, because our motivation is the compiled setting, we will focus on the partial models on $\mcH_B$. Furthermore, we remark that the notion of pure partial models is not vacuous. In particular, the optimal bipartite model for the CHSH inequality has a pure partial model on $\mcH_B$ \cite{CHSH}. With the notion of a partial quantum model, we define the notion of a partial (or one-sided) self-test for a bipartite model.

\begin{definition}[Partial self-test]\label{def:weakst}
    The inequality $I \leq \eta^{\mathrm{Q}}$ is a partial self-test for a bipartite model $\mathrm{Q}^*= \left( \lbrace P_{a|x} \rbrace , \lbrace Q_{b|y} \rbrace, \ket{\phi} \right)$ with a pure partial model $\left( \lbrace Q_{b|y} \rbrace, \lbrace \ket{\phi_{a|x}}\rbrace \right)$ if there exists a non-negative function $f(\epsilon)$ such that $f(\epsilon) \to 0$ as $\epsilon \to 0$, such that for any partial quantum model $\mathrm{Q} =  \left( \lbrace N_{b|y} \rbrace, \lbrace \rho_{a|x} \rbrace \right) $ achieving $I \geq \eta^{\mathrm{Q}} - \epsilon$ for $\epsilon \geq 0$, there exist  a Hilbert space $\mathcal{H}_{\mathrm{aux}}$, a collection of auxiliary states $\{\sigma_{a|x}\}$ and an isometry $V:\mcH_B\to \mathbb{C}^{d} \otimes \mathcal{H}_{\mathrm{aux}}$ such that the following is satisfied for all $x,y,a,b$:
\begin{equation*}
\begin{aligned}
        \big \| V N_{b |y}\rho_{a|x} N_{b|y} V^{\dagger} - Q_{b|y}\ketbra{\phi_{a|x}}{\phi_{a|x}}Q_{b|y} \otimes \sigma_{a|x} \big \|_{2} &\leq f(\epsilon) \\ \text{and} \quad
        \big \| V \rho_{a|x} V^{\dagger} - \ketbra{\phi_{a|x}}{\phi_{a|x}} \otimes \sigma_{a|x} \big \|_{2} &\leq f(\epsilon), \hspace{0.2cm} 
\end{aligned}
\end{equation*}
\end{definition}

Give the symmetry of $\mcH_A$ and $\mcH_B$ in the bipartite case, one can define a notion of partial self-test for either subsystem. Given a bipartite self-test, one can check that tracing out either subsystem results in a partial self-test. We leave it as an open question as to whether a partial self-test (say over $\mcH_A$ and over $\mcH_B$) implies that the correlation is a bipartite self-test.

\subsection{Compiled self-tests from partial models}
\label{sec:comp_self_test}

 There are two main difficulties with self-testing in the compiled setting. Firstly, the \emph{correctness with respect to auxiliary systems} property of the compiler (see \emph{Property (1)} in \cref{def_qhe}) only guarantees that a QPT prover can prepare states (possibly mixed) $\rho_{a|x}$ over $\mcH_B$ that are negligible in trace distance from the post measurement states $P_{a|x}\ketbra{\Psi}{\Psi}P_{a|x}/p(a|x)$ of the ideal bipartite model $\mathrm{Q}$. This puts a fundamental constraint on our ability to exactly describe the set of ideal models in the compiled setting. Secondly, unlike in the nonlocal setting, it is not clear how to extract information about the measurements and states in the first round due to the homomorphic evaluation of the measurements and preparation of the states. To address these challenges we introduce the compiled counter-part of a partial quantum model.

 Recall that a compiled (quantum) model $\widetilde{\mathrm{Q}}$ consists of a family of post-measurement states for ``Alice'' $\ket{\widetilde{\phi}_{\alpha | \chi}}$, which correspond to the state of the device following the encrypted question $\chi$, and encrypted answer $\alpha$, and a POVM $\{N_{b|y}\}$ employed by ``Bob''. One could also consider a more general compiled quantum model, which includes a description of the initial state and Alice's operators. The point of taking the coarser model is that it allows us to introduce the notion of the \emph{compiled-counterpart} of a bipartite model $\mathrm{Q}$, which relates the post-measurement information in the bipartite setting with another bipartite model that resembles a compiled model.

\begin{definition}[Compiled-counterpart model]\label{Def:comp_counter)}
    Given a pure partial model $\mathrm{Q}'$, the compiled-counterpart model of $\mathrm{Q}'$ is the pure partial model $\widetilde{\mathrm{Q}}^{(\lambda)}= ( \{\ket{\widetilde{\phi}_{\alpha | \chi}^{(\lambda)}} \}, \{Q_{b|y}^{(\lambda)}\})$
    satisfying the following conditions for all $\lambda\in \N$:
\begin{align*}
   \ket{\widetilde{\phi}_{\alpha | \chi}^{(\lambda)}} &= \ket{\phi_{a|x}}, \ \text{ for all } \ \sk: \gen(1^\lambda)=\sk , \ \chi : \enc(x,\sk) = \chi,  \ \alpha : \dec(\alpha,\sk) = a.\\ N_{b|y}^{(\lambda)} &= Q_{b|y}, \ \text{ for all } b, y.
\end{align*}
\end{definition}

We remark that the compiled counterpart need not be an actual compiled model. For example, it is not required to satisfy the QPT conditions needed of a compiled model. Instead it is a model that resembles an idealized version of an honest implementation of a partial model under homomorphic encryption. We proceed with a definition of self-testing in the compiled setting that resembles partial self-testing in the bipartite setting in the context of these compiled-counterparts.

\begin{definition}[Compiled self-test]\label{def:compst}
    Let $I$ denote a Bell expression with an optimal pure partial model $\mathrm{Q}^*$. The inequality $I \leq \eta^{\mathrm{Q}}$ is a compiled self-test for the corresponding compiled-counterpart $\widetilde{\mathrm{Q}}^* = ( \lbrace \ket{\widetilde{\phi}_{\alpha|\chi}} \rbrace , \lbrace Q_{b|y} \rbrace )$, if there exists a non-negative function $f(\epsilon)$ such that $f(\epsilon) \to 0$ as $\epsilon \to 0$, such that for every pure and projective compiled model $\widetilde{\qmodel} = \left( \lbrace \ket{\Psi_{\alpha|\chi}} \rbrace , \lbrace N_{b|y} \rbrace \right)$ that achieves $I \geq \eta^{\mathrm{Q}} - \epsilon$ for some $\epsilon \geq 0$, there exists a negligible function $\negl(\lambda)$, an isometry $V:\tilde{\mathcal{H}} \to \mathbb{C}^{d} \otimes \mathcal{H}_{\mathrm{aux}}$, and auxiliary states $\ket{\mathsf{aux}_{\alpha|\chi}} \in \mathcal{H}_{\mathrm{aux}}$, which satisfy the following for all $x,b, y$:
 \begin{subequations}
 \begin{eqnarray}
            \expect_{\chi : \enc(x) = \chi}\sum_{\alpha}\big \| V  \ket{\Psi_{\alpha|\chi}} - \ket{\widetilde{\phi}_{\alpha | \chi}} \otimes \ket{\mathsf{aux}_{\alpha|\chi}} \big \|^{2} \leq \negl(\lambda) + f(\epsilon), \hspace{0.2cm} \text{and} \label{eq:stdef_a}\\
            \expect_{\chi : \enc(x) = \chi}\sum_{\alpha}\big \| V  N_{b|y} \ket{\Psi_{\alpha|\chi}} - Q_{b|y}\ket{\widetilde{\phi}_{\alpha | \chi}} \otimes \ket{\mathsf{aux}_{\alpha|\chi}} \big \|^{2} \leq \negl(\lambda) + f(\epsilon). \label{eq:stdef_b}
\end{eqnarray} \label{eq:stdef}
\end{subequations} \label{def:compst_inf}
\end{definition}

Equation \eqref{eq:stdef_a} is a statement about the provers state after the first round. It asserts that, given a question $x$ and answer $a$, the post-measurement state is negligibly close to that of an ideal prover implementing the honest bipartite model. To see this concretely, suppose the right hand side was exactly equal to zero. Then we have the equality $V\ket{\Psi_{\alpha|\chi}} = \ket{\widetilde{\phi}_{\alpha | \chi}} \otimes \ket{\mathsf{aux}_{\alpha|\chi}}$ for all $\chi$ such that $\enc(x) = \chi$ and all $\alpha$. Substituting $\ket{\widetilde{\phi}_{\alpha | \chi}}$ for the states $\ket{\phi_{a|x}}$ from Definition \ref{Def:comp_counter)}, we obtain
\begin{equation}
    V\ket{\Psi_{\alpha|\chi}} = \ket{\phi_{a|x}}\otimes \ket{\mathsf{aux}_{\alpha|\chi}} \label{eq:example_s}
\end{equation}
whenever $\enc(x) = \chi$ and $\dec(\alpha) = a$. That is, the post-measurement states are equal to the target states up an isometry. Therefore, we interpret \eqref{eq:stdef_a} as an approximate version of \cref{eq:example_s}, which accounts for a finite size security parameter $\lambda$ and small errors in the Bell violation $\epsilon$. Equation \eqref{eq:stdef_b} is the analogous statement including the measurements in the second round. We remark that if $V$ could depend on the question $x$ and answer $a$, \eqref{eq:stdef_a} would trivially hold regardless of the compiled Bell violation, since the states $\ket{\phi_{a|x}}$ could be prepared directly. It is therefore essential to enforce the same isometry is applied for all $a$ and $x$. Furthermore, \eqref{eq:stdef_b} captures several existing self-testing results in the compiled setting. For example those presented in~\cite[Lemma 34]{NZ23},~\cite[Theorem 3.6]{Cui24} and~\cite[Eqs. 98 and 103]{Baroni24}. Our proposed definition then goes further by also certifying the states after the first round but before Bob's measurements, as captured by \eqref{eq:stdef_a}. 

It is natural to ask if Definition \ref{def:compst} is the strongest form of self-testing possible in this scenario, or if one can also certify the initial state $\ket{\Psi}$ before Alice's measurements. An initial guess would be to show there exists an isometry $V$ satisfying
\begin{equation}
    V\ket{\Psi} \approx_{\negl(\lambda)} \ket{\phi} \otimes \ket{\mathsf{aux}},
\end{equation}
where $\ket{\phi}$ is the ideal bipartite entangled state. However, on its own this statement is not very useful: such an isometry always exists, namely, one which ignores $\ket{\Psi}$ and prepares $\ket{\phi}$ directly. A possible way around is to demand the same $V$ also satisfies \eqref{eq:stdef}. At a glance, this suggests certifying the initial state alone is not meaningful in the single prover setting; one always needs to also consider the measurements. This contrasts the two prover setting, where self-testing statements made only about the state are known~\cite{SupicSelfTest} and non-trivial due to the space-like separation of the provers. Another question worth asking is if the assumption of having a pure projective models $\widetilde{\qmodel}$ can be relaxed in the definition \cref{def:compst}.

\subsection{Compiled self-test for tilted-CHSH inequalities}

Our final result is that the extended tilted-CHSH Bell inequalities are compiled self-tests according to \cref{def:compst}. In particular, we have the following result.

\begin{theorem}\label{thm:tilt_st}
Let $\theta \in (0,\pi/4]$, $\phi \in \big( \max\{-2\theta,-\pi + 2\theta\} , \min\{2\theta,\pi - 2\theta\}\big)  \setminus \{0\} $, and let $I_{\theta,\phi}$ be the generalized tilted-CHSH functional with quantum bound $ \eta^{\mathrm{Q}}_{\theta,\phi}$ according to \cref{eq:tiltBop}. Then the inequality $I_{\theta,\phi} \leq \eta^{\mathrm{Q}}_{\theta,\phi}$ is a compiled self-test for the compiled-counter part of \cref{eq:honCHSH} according to \cref{def:compst}.
\end{theorem}

\noindent The proof is reminiscent of the approach in \cite{BampsPironio}, and includes similar calculations to those used in~\cite{NZ23,Baroni24} which establish rigidity statements in the compiled setting. The remainder of this section is devoted to proving the above theorem. Recall the SOS decomposition for $\bar{S}_{\theta,\phi}$, given by the polynomials $N_{0},N_{1}$ defined in \cref{eq:SOSpoly}. Then for compiled models $I_{\theta,\phi} = \pseudo[S_{\theta,\phi}] \geq \eta^{\mathrm{Q}}_{\theta,\phi}  - \epsilon$, for some $\epsilon \geq 0$, implies $\pseudo[\bar{S}_{\theta,\phi}] = \eta^{\mathrm{Q}}_{\theta,\phi} - \pseudo[S_{\theta,\phi}] \leq \epsilon$. Using the calculation in \cref{eq:bigCalc} with $P = N_{0}$ and $P = N_{1}$ respectively, we find
    \begin{equation}
    \begin{aligned}
        \pseudo[N_{0}^{\dagger}N_{0}] &\approx_{\negl(\lambda)} \expect_{\chi:\enc(x=0)=\chi} \sum_{\alpha}\bra{\Psi_{\alpha|\chi}}\Big|(-1)^{\dec(\alpha)}\Id - Z_{B}\Big|^{2}\ket{\Psi_{\alpha|\chi}} =: \hat{E}_{0}, \\
        \pseudo[N_{1}^{\dagger}N_{1}] &\approx_{\negl(\lambda)} \expect_{\chi:\enc(x=1)=\chi} \sum_{\alpha}\bra{\Psi_{\alpha|\chi}}\Big|\Id - (-1)^{\dec(\alpha)}\sin(2\theta)X_{B} - \cos(2\theta)Z_{B}\Big|^{2}\ket{\Psi_{\alpha|\chi}} =: \hat{E}_{1},
    \end{aligned}
    \end{equation}
    where we defined $Z_{B} := (B_{0} + B_{1})/(2\cos(\phi))$ and $X_{B} := (B_{0} - B_{1})/(2\sin(\phi))$, and note that $\hat{E}_{i} \geq 0$ for $i=0,1$. This implies $\pseudo[N_{i}^{\dagger}N_{i}] \geq \hat{E}_{i} - \negl(\lambda)$, from which we observe that 
    \begin{equation}
        \hat{E}_{0} + \tau^{2}_{\theta,\phi} \hat{E}_{1} \leq \pseudo[N_{0}^{\dagger}N_{0}] + \tau^{2}_{\theta,\phi} \pseudo[N_{1}^{\dagger}N_{1}] + \negl(\lambda)' = \pseudo[\bar{S}_{\theta,\phi}] + \negl(\lambda)' \leq \epsilon + \negl(\lambda)',
    \end{equation}
    where $\negl(\lambda)' = (1+\tau^{2}_{\theta,\phi})\negl(\lambda)$. Since $\hat{E}_{i} \geq 0$, we arrive at 
    \begin{equation}
        \hat{E}_{0} \leq \epsilon + \negl(\lambda)' \ \ \mathrm{and} \ \ \tau^{2}_{\theta,\phi} \hat{E}_{1} \leq \epsilon + \negl(\lambda)'. \label{eq:sosConst}
    \end{equation}

    The next part of the proof proceeds as follows: using each constraint in \cref{eq:sosConst}, we deduce the (approximate) algebraic structure of the operators $Z_{B}$ and $X_{B}$ on the states $\ket{\Psi_{\alpha|\chi}}$. Namely, we establish anti-commutation relations and approximate unitary relations. We break the proof up into two sections: extracting the states, and extracting the measurements. 

    \paragraph*{Extracting the states}

    We begin by studying the relations imposed by $\hat{E}_{0} \leq \epsilon + \negl(\lambda)'$. Define $\delta_{0} := \epsilon + \negl(\lambda)'$. We begin with several claims:
    \begin{claim}
        $\hat{E}_{0} = \expect_{\chi:\enc(x=0)=\chi} \sum_{\alpha}\big \| \big((-1)^{\dec(\alpha)}\Id - Z_{B}\big)\ket{\Psi_{\alpha|\chi}}\big\|^{2} \leq \delta_{0}$. \label{claim:1}
    \end{claim}

    \begin{claim}
    $\expect_{\chi:\enc(x=0)=\chi} \sum_{\alpha}\big \| \big(\Id - (Z_{B})^{2}\big)\ket{\Psi_{\alpha|\chi}}\big\|^{2} \leq \delta_{0}\big( 1 + 1/\cos(\phi)\big)^{2} =: \delta_{1}$. \label{claim:2}
\end{claim}
\begin{proof}
    Note that
    \begin{equation*}
        \begin{aligned}
            \big \| \big(\Id - (Z_{B})^{2}\big)\ket{\Psi_{\alpha|\chi}}\big\|
            &\leq \big \| Z_{B}\big((-1)^{\dec(\alpha)}\Id - Z_{B}\big) \ket{\Psi_{\alpha|\chi}} \big \| \\
            & \ \ \ \ + \big \|  \big( (-1)^{\dec(\alpha)}\Id - Z_{B}\big)\ket{\Psi_{\alpha|\chi}}\big\| \\
            &\leq (1+\| Z_{B} \|_{\mathrm{op}}) \big \| \big((-1)^{\dec(\alpha)}\Id - Z_{B}\big) \ket{\Psi_{\alpha|\chi}} \big \|.
        \end{aligned}
    \end{equation*}
    In the first inequality, we used the triangle inequality, and for the second inequality, we used the definition of the operator norm $\|A\ket{\psi}\| \leq \| A \|_{\mathrm{op}} \| \ket{\psi} \|, \ \ \forall\; \ket{\psi} \in \mathcal{H}$. By the definition of $Z_{B}$, $\| Z_{B} \|_{\mathrm{op}} \leq (1/(2\cos(\phi)))( \|B_{0}\|_{\mathrm{op}} + \|B_{1}\|_{\mathrm{op}}) = 1/\cos(\phi)$. Squaring both sides, taking the expectation over $\chi:\enc(x=0)=\chi$ and summing over $\alpha$, we can apply \cref{claim:1} to get the desired bound. 
\end{proof}
\begin{claim}
    $\expect_{\chi:\enc(x=0)=\chi} \sum_{\alpha}\big \| \big(2\cos(2\phi)\Id - \{B_{0},B_{1}\}\big)\ket{\Psi_{\alpha|\chi}}\big\|^{2} \leq 16\cos^{4}(\phi)\delta_{1} =: \delta_{2}$. \label{claim:3}
\end{claim}
\begin{proof}
    Note that $(Z_{B})^{2} = (2 \Id + \{B_{0},B_{1}\})/(4\cos^{2}(\phi))$, and therefore
    \begin{equation*}
        \begin{aligned}
            \big \| \big( \Id - (Z_{B})^{2} \big) \ket{\Psi_{\alpha |\chi}} \big \|^{2} &= \big \| \big( \Id - (2 \Id + \{B_{0},B_{1}\})/(4\cos^{2}(\phi)) \big) \ket{\Psi_{\alpha |\chi}} \big \|^{2} \\
            &= \frac{1}{16\cos^{4}(\phi)}\big \| \big( 2\cos(2\phi)\Id -\{B_{0},B_{1}\} \big) \ket{\Psi_{\alpha |\chi}} \big \|^{2}.
        \end{aligned}
    \end{equation*}
    Taking the expectation over $\chi:\enc(x=0)=\chi$ and summing over $\alpha$, followed by applying \cref{claim:2}, establishes the desired bound. 
\end{proof}
\begin{claim}
    $\expect_{\chi:\enc(x=0)=\chi} \sum_{\alpha}\big \| \big(\Id - (X_{B})^{2}\big)\ket{\Psi_{\alpha|\chi}}\big\|^{2} \leq  \delta_{1}/\tan^{4}(\phi) =: \delta_{3}$.
\end{claim}
\begin{proof}
    Noting $(X_{B})^{2} = (2 \Id - \{B_{0},B_{1}\})/(4\sin^{2}(\phi))$,
    \begin{equation*}
        \begin{aligned}
            \big \| \big( \Id - (X_{B})^{2} \big) \ket{\Psi_{\alpha |\chi}} \big \|^{2} &= \big \| \big( \Id - (2 \Id - \{B_{0},B_{1}\})/(4\sin^{2}(\phi)) \big) \ket{\Psi_{\alpha |\chi}} \big \|^{2} \\
            &= \frac{1}{16\sin^{4}(\phi)}\big \| \big( 2\cos(2\phi)\Id -\{B_{0},B_{1}\} \big) \ket{\Psi_{\alpha |\chi}} \big \|^{2}.
        \end{aligned}
    \end{equation*}
    We obtain the claim by taking the expectation over $\chi:\enc(x=0)=\chi$, summing over $\alpha$ and applying \cref{claim:3}.
\end{proof}

To define the isometry, we need to define unitary operators $\tilde{Z}_{B}$ and $\tilde{X}_{B}$ from $Z_{B}$ and $X_{B}$, which maintain the above relations. For this, we follow the \emph{regularization} procedure of~\cite{YangSelfTest,MYS12,BampsPironio}: define $Z_{B}^{*}$ from $Z_{B}$ by changing all zero eigenvalues of $Z_{B}$ to 1, and define the operator $\tilde{Z}_{B} := Z_{B}^{*} |Z_{B}^{*}|^{-1}$. Then $\tilde{Z}_{B}$ is unitary, self-adjoint, commutes with $Z_{B}$ and satisfies $\tilde{Z}_{B}Z_{B} = |Z_{B}|$. We now establish the following:
\begin{claim}
    $\expect_{\chi:\enc(x=0)=\chi} \sum_{\alpha}\big \| \big(\tilde{Z}_{B} - Z_{B}\big)\ket{\Psi_{\alpha|\chi}}\big\|^{2} \leq \delta_{0}$. \label{claim:5}
\end{claim}
\begin{proof}
    Using the fact that $\tilde{Z}_{B}^{\dagger}$ is unitary and $\tilde{Z}_{B}Z_{B} = |Z_{B}|$,
    \begin{equation*}
        \| \big(\tilde{Z}_{B} - Z_{B}\big)\ket{\Psi_{\alpha|\chi}}\big\|^{2} = \| \big(\Id - |(-1)^{\dec(\alpha)}Z_{B}|\big)\ket{\Psi_{\alpha|\chi}}\big\|^{2}.
    \end{equation*}
    For any self-adjoint operator $A$, we have $\| (\Id - |A|)\ket{\psi} \|^{2} = \bra{\psi}(\Id + |A|^{2} - 2|A|)\ket{\psi} \leq \bra{\psi}(\Id + |A|^{2} - 2A)\ket{\psi} = \|(\Id - A)\ket{\psi}\|^{2}$. Applying this with $A = (-1)^{\dec(\alpha)}Z_{B}$, taking the expectation over $\chi:\enc(x=0)=\chi$ and summing over $\alpha$, we can directly applying \cref{claim:1} to get the desired bound.
\end{proof}

\begin{claim}
    $\expect_{\chi:\enc(x=0)=\chi} \sum_{\alpha}\big \| \big(\tilde{X}_{B} - X_{B}\big)\ket{\Psi_{\alpha|\chi}}\big\|^{2} \leq \delta_{1}/\tan^{2}(\phi) =: \delta_{4}$. \label{claim:6}
\end{claim}
\begin{proof}
     Note that
     \begin{equation*}
         \begin{aligned}
             \| \big(\tilde{X}_{B} - X_{B}\big)\ket{\Psi_{\alpha|\chi}}\big\|^{2} &= \| \big(\Id - |X_{B}|\big)\ket{\Psi_{\alpha|\chi}}\big\|^{2}\\
             &\leq \| \big(\Id + |X_{B}|\big)\big(\Id - |X_{B}|\big)\ket{\Psi_{\alpha|\chi}}\big\|^{2},
         \end{aligned}
     \end{equation*}
     where we used the same steps in \cref{claim:5} for the first equality, and the operator inequality $\Id + |X_{B}| \geq \Id$ for the inequality. Notice that $|X_{B}|^{2} = (X_{B})^{2}$, which satisfies $\cos^{2}(\phi)(Z_{B})^{2} + \sin^{2}(\phi)(X_{B})^{2} = \Id$ by definition. Inserting this into the above expression we find that
     \begin{equation*}
         \begin{aligned}
             \| \big(\Id + |X_{B}|\big)\big(\Id - |X_{B}|\big)\ket{\Psi_{\alpha|\chi}}\big\|^{2} &= \frac{1}{\tan^{2}(\phi)}\| \big(\Id - (Z_{B})^{2}\big)\ket{\Psi_{\alpha|\chi}}\big\|^{2},
         \end{aligned}
     \end{equation*}
     and taking the expectation over $\chi:\enc(x=0)=\chi$, summing over $\alpha$ and applying \cref{claim:2} we obtain the desired bound. 
\end{proof}
Next, we define the operators $P_{B}^{b} = \frac{1}{2}(\Id + (-1)^{b}\tilde{Z}_{B})$ for $b \in \{0,1\}$
and note a key property.
\begin{claim}
    $\expect_{\chi:\enc(x=0)=\chi} \sum_{\alpha}\big \| \big(P_{B}^{b} - \delta_{\dec(\alpha),b}\Id\big)\ket{\Psi_{\alpha|\chi}}\big\|^{2} \leq \delta_{0}$. \label{claim:7}
\end{claim}
\begin{proof}
    Note that
    \begin{equation*}
        \begin{aligned}
            \| \big(P_{B}^{b} - \delta_{\dec(\alpha),b}\Id\big)\ket{\Psi_{\alpha|\chi}}\big\| &= \frac{1}{2}\|\big( \tilde{Z}_{B} - (-1)^{\dec(\alpha)}\Id\big)\ket{\Psi_{\alpha|\chi}}\big\|,
        \end{aligned}
    \end{equation*}
    where we used the definition of $P_{B}^{b}$ and  $\delta_{\dec(\alpha),b} = \frac{1}{2}( 1 + (-1)^{\dec(\alpha) + b})$. Applying the triangle inequality,
    \begin{equation*}
        \|\big( \tilde{Z}_{B} - (-1)^{\dec(\alpha)}\Id\big)\ket{\Psi_{\alpha|\chi}}\big\| \leq \|\big( Z_{B} - (-1)^{\dec(\alpha)}\Id\big)\ket{\Psi_{\alpha|\chi}}\big\| + \|\big( \tilde{Z}_{B} - Z_{B}\big)\ket{\Psi_{\alpha|\chi}}\big\|.
    \end{equation*}
    Next, we square both sides of the above, and use the fact that for two real numbers $a,b$, $(a+b)^{2}\leq 2(a^{2} + b^{2})$, to obtain 
    \begin{equation}
        \|\big( \tilde{Z}_{B} - (-1)^{\dec(\alpha)}\Id\big)\ket{\Psi_{\alpha|\chi}}\big\|^{2} \leq 2\|\big( Z_{B} - (-1)^{\dec(\alpha)}\Id\big)\ket{\Psi_{\alpha|\chi}}\big\|^{2} + 2\|\big( \tilde{Z}_{B} - Z_{B}\big)\ket{\Psi_{\alpha|\chi}}\big\|^{2}.
    \end{equation}
    By taking the expectation over $\chi \, : \, \enc(x=0)=\chi$ and summing over $\alpha$, the first term can be bounded by applying \cref{claim:1}, while the second can be bounded by applying \cref{claim:5}. 
\end{proof}

Having established the above relations for the case $\chi = \enc(x=0)$, we are now ready to define the isometry $V:\mathcal{H} \rightarrow \mathbb{C}^{2} \otimes \mathcal{H}$,
\begin{equation}
    V = \sum_{b \in \{0,1\}} \ket{b} \otimes (\tilde{X}_{B})^{b} P_{B}^{b}. \label{eq:extr}
\end{equation}
This isometry is based on the ``SWAP gate isometry''~\cite{MYS12,YangSelfTest}.

Using \cref{claim:7}, we establish the following lemma:
\begin{lemma}
    Let $\widetilde{\mathrm{Q}}$ be any compiled model which satisfies $\pseudo[S_{\theta,\phi}] \geq \eta_{\theta,\phi}^{\mathrm{Q}} - \epsilon$, and let $V$ be the isometry defined in \cref{eq:extr}. Then there exists auxiliary states $\{\ket{\mathsf{aux}_{\alpha|\chi}}\}_{\alpha,\chi}$ such that
    \begin{equation}
        \expect_{\chi:\enc(x=0)=\chi} \sum_{\alpha}\big \| V\ket{\Psi_{\alpha|\chi}} - \ket{\dec(\alpha)} \otimes \ket{\mathsf{aux}_{\alpha|\chi}} \big\|^{2} \leq 2\delta_{0}.
    \end{equation} \label{lem:st1}
\end{lemma}
\begin{proof}
    Consider the states $\ket{\mathsf{aux}_{\alpha|\chi}} = (\tilde{X}_{B})^{\dec(\alpha)}\ket{\Psi_{\alpha|\chi}}$. Then 
    \begin{equation*}
        \begin{aligned}
            \ket{\dec(\alpha)} \otimes \ket{\mathsf{aux}_{\alpha|\chi}} &= \sum_{b \in \{0,1\}} \delta_{b,\dec(\alpha)}\ket{b} \otimes (\tilde{X}_{B})^{b}\ket{\Psi_{\alpha|\chi}},
        \end{aligned}
    \end{equation*}
    and we have
    \begin{equation*}
        \begin{aligned}
            \big\|V\ket{\Psi_{\alpha|\chi}} - \ket{\dec(\alpha)} \otimes \ket{\mathsf{aux}_{\alpha|\chi}} \big \|^{2} &= \big \| \sum_{b \in \{0,1\}}\ket{b} \otimes (\tilde{X}_{B})^{b}\big( P_{B}^{b} - \delta_{b,\dec(\alpha)} \Id \big) \ket{\Psi_{\alpha|\chi}} \big \|^{2} \\
            &\leq \sum_{b \in \{0,1\}} \big \| \big( P_{B}^{b} - \delta_{b,\dec(\alpha)} \Id \big) \ket{\Psi_{\alpha|\chi}} \big \|^{2}.
        \end{aligned}
    \end{equation*}
     By taking the expectation over $\chi:\enc(x=0)=\chi$, and summing over $\alpha$, we can apply \cref{claim:7} to obtain the desired bound.  
\end{proof}

 We now proceed to study the relations imposed by the second SOS polynomial, which imply $\hat{E}_{1} \leq \delta_{0}/\tau_{\theta,\phi}^{2}$. Before we begin, we note that in any earlier claim that did not include a factor of $(-1)^{\dec(\alpha)}$, we can always apply \cref{lem:QHE1} to exchange the expectation over encryptions of $x=0$ to an expectation over encryptions of $x = 1$ at the cost of adding $\negl(\lambda)$ to the bound. This follows from the fact that such claims are bounding the expectation of operators $\mathcal{N}$ constructed from sums of products of $B_{0}$ and $B_{1}$ only along with the fact that $\expect_{\chi:\enc(x=0)=\chi}\sum_{\alpha}\bra{\Psi_{\alpha|\chi}}\mathcal{N}\ket{\Psi_{\alpha|\chi}} \approx_{\negl(\lambda)} \expect_{\chi:\enc(x=1)=\chi}\sum_{\alpha}\bra{\Psi_{\alpha|\chi}}\mathcal{N}\ket{\Psi_{\alpha|\chi}}$.  
 \begin{claim}
     $\hat{E}_{1} = \expect_{\chi:\enc(x=1)=\chi} \sum_{\alpha}\big\| \big(\Id - (-1)^{\dec(\alpha)}\sin(2\theta)X_{B} - \cos(2\theta)Z_{B}\big)\ket{\Psi_{\alpha|\chi}}\big\|^{2} \leq \delta_{0}/\tau_{\theta,\phi}^{2}$. \label{claim:8}
 \end{claim}

 \begin{claim}
     $\expect_{\chi:\enc(x=1)=\chi} \sum_{\alpha}\big\| \big(\Id - (-1)^{\dec(\alpha)}\sin(2\theta)\tilde{X}_{B} - \cos(2\theta)\tilde{Z}_{B}\big)\ket{\Psi_{\alpha|\chi}}\big\|^{2} \leq \delta_{5}$, where $\delta_{5} := 2\delta_{0}/\tau_{\theta,\phi}^{2} + 4\sin^{2}(2\theta)(\delta_{4} + \negl(\lambda)) + 4\cos^{2}(2\theta)(\delta_{0} + \negl(\lambda))$. \label{claim:9}
 \end{claim}
 \begin{proof}
     Using the triangle inequality we obtain
     \begin{equation*}
         \begin{aligned}
             \big\| \big(\Id -& (-1)^{\dec(\alpha)}\sin(2\theta)\tilde{X}_{B} - \cos(2\theta)\tilde{Z}_{B}\big)\ket{\Psi_{\alpha|\chi}}\big\| \\
             &\leq \big\| \big(\Id - (-1)^{\dec(\alpha)}\sin(2\theta)X_{B} - \cos(2\theta)Z_{B}\big)\ket{\Psi_{\alpha|\chi}}\big\| \\
             &+ \big\| \sin(2\theta)\big(X_{B} - \tilde{X}_{B}\big)\ket{\Psi_{\alpha|\chi}}\big\| + \big\| \cos(2\theta)\big(Z_{B} - \tilde{Z}_{B}\big)\ket{\Psi_{\alpha|\chi}}\big\|.
         \end{aligned}
     \end{equation*}
     Next, we square both sides, and use the fact that, for three real numbers $a,b,c$, $(a+b+c)^{2}\leq 4(a^{2}+b^{2})+2c^{2}$, to arrive at 
     \begin{equation*}
         \begin{aligned}
             \big\| \big(\Id -& (-1)^{\dec(\alpha)}\sin(2\theta)\tilde{X}_{B} - \cos(2\theta)\tilde{Z}_{B}\big)\ket{\Psi_{\alpha|\chi}}\big\|^{2} \\
             &\leq 2\big\| \big(\Id - (-1)^{\dec(\alpha)}\sin(2\theta)X_{B} - \cos(2\theta)Z_{B}\big)\ket{\Psi_{\alpha|\chi}}\big\|^{2} \\
             &+ 4\big\| \sin(2\theta)\big(X_{B} - \tilde{X}_{B}\big)\ket{\Psi_{\alpha|\chi}}\big\|^{2} + 4\big\| \cos(2\theta)\big(Z_{B} - \tilde{Z}_{B}\big)\ket{\Psi_{\alpha|\chi}}\big\|^{2}.
         \end{aligned}
     \end{equation*}
     We now take the expectation over $\chi:\enc(x=1)=\chi$, and sum over $\alpha$. This allows us to apply \cref{claim:8} to the first term, \cref{lem:QHE1} followed by \cref{claim:6} to the second, and \cref{lem:QHE1} followed by \cref{claim:5} to the final term, giving the desired bound. 
 \end{proof}
\begin{claim}
    $\expect_{\chi:\enc(x=1)=\chi} \sum_{\alpha}\big \| \{\tilde{Z}_{B},\tilde{X}_{B}\}\ket{\Psi_{\alpha|\chi}}\big\|^{2} \leq \delta_{6}$, where $\delta_{6} := 2(\delta_{0} + \negl(\lambda)) + 8(\delta_{4} + \negl(\lambda)) + 8(\delta_{0} + \negl(\lambda))/(2\sin^{2}(\phi)) + 16(1 + 1/(2\cos^{2}(\phi)))\delta_{0}/(\sin(2\theta)\tau_{\theta,\phi})^{2} + 16(1/\sin^{2}(2\theta) + 1/(2\cos^{2}(\phi)\tan^{2}(2\theta)))(\delta_{0} + \negl(\lambda))$. \label{claim:10}
\end{claim}
\begin{proof}
    We begin by applying the triangle inequality to obtain
    \begin{equation*}
        \big \| \{\tilde{Z}_{B},\tilde{X}_{B}\} \ket{\Psi_{\alpha|\chi}}\big \|^{2} \leq 2\big \| \tilde{X}_{B}\big( \tilde{Z}_{B} - (-1)^{\dec(\alpha)}\Id\big)\ket{\Psi_{\alpha|\chi}} \big \|^{2} + 2\big \| \big( \tilde{Z}_{B} + (-1)^{\dec(\alpha)}\Id\big) \tilde{X}_{B}\ket{\Psi_{\alpha|\chi}}\big \|^{2}.
    \end{equation*}
    For the first term, we can use the fact that $\tilde{X}_{B}$ is unitary, and when taking the expectation over $\chi:\enc(x=1)=\chi$ and sum over $\alpha$, we can apply \cref{lem:QHE1} followed by \cref{claim:1} to obtain a bound of $2(\delta_{0} + \negl(\lambda))$. For the second term, 
    \begin{equation*}
    \begin{aligned}
        2\big \| \big( \tilde{Z}_{B} + (-1)^{\dec(\alpha)}\Id\big)\tilde{X}_{B}\ket{\Psi_{\alpha|\chi}} \big \|^{2}  &\leq 4\big \| \big(\tilde{Z}_{B} + (-1)^{\dec(\alpha)}\Id\big)(\tilde{X}_{B} - X_{B})\ket{\Psi_{\alpha|\chi}}\big \|^{2} \\
        &+ 4\| \big(\tilde{Z}_{B} + (-1)^{\dec(\alpha)}\Id\big) X_{B}\ket{\Psi_{\alpha|\chi}}\big \|^{2} \\
        &\leq 4\big \| \tilde{Z}_{B} + (-1)^{\dec(\alpha)}\Id\|_{\mathrm{op}}^{2}\big \| (\tilde{X}_{B} - X_{B})\ket{\Psi_{\alpha|\chi}}\big \|^{2} \\
        &+ 8\| \big(Z_{B} + (-1)^{\dec(\alpha)}\Id\big) X_{B}\ket{\Psi_{\alpha|\chi}}\big \|^{2} \\
        &+ 8\| \big(\tilde{Z}_{B} - Z_{B}\big) X_{B}\ket{\Psi_{\alpha|\chi}}\big \|^{2}.
    \end{aligned}
    \end{equation*}
     Since $\tilde{Z}_{B}$ is unitary $\big \| \tilde{Z}_{B} + (-1)^{\dec(\alpha)}\Id\|_{\mathrm{op}}^{2} \leq 2$, hence when taking the expectation over $\chi:\enc(x=1)=\chi$ and sum over $\alpha$ for the first term above, we can apply \cref{lem:QHE1} followed by \cref{claim:6} to $\| (\tilde{X}_{B} - X_{B})\ket{\Psi_{\alpha|\chi}} \|^{2}$; the product results in a contribution $8(\delta_{4} + \negl(\lambda))$. Next, we note that
    \begin{equation*}
    \begin{aligned}
        8\big \|\big(Z_{B} + (-1)^{\dec(\alpha)}\Id\big) X_{B} \ket{\Psi_{\alpha|\chi}} \big \|^{2} &= 8\big \| X_{B}\big(Z_{B} - (-1)^{\dec(\alpha)}\Id\big)\ket{\Psi_{\alpha|\chi}} \big \|^{2} \\
        &\leq 8\|X_{B}\|_{\mathrm{op}}^{2} \|\big(Z_{B} - (-1)^{\dec(\alpha)}\Id\big)\ket{\Psi_{\alpha|\chi}} \big \|^{2},
    \end{aligned}
    \end{equation*}
    where the equality follows since $\{Z_{B},X_{B}\} = 0$ by definition. Now $\|X_{B}\|_{\mathrm{op}}^{2} \leq 1/(2\sin^{2}(\phi))$, and when taking the expectation over $\chi:\enc(x=1)=\chi$ and sum over $\alpha$, we can apply \cref{lem:QHE1} followed by \cref{claim:1}, resulting in a contribution of $8(\delta_{0} + \negl(\lambda))/(2\sin^{2}(\phi))$. 

    We now address the term $8\| \big(\tilde{Z}_{B} - Z_{B}\big) X_{B}\ket{\Psi_{\alpha|\chi}}\big \|^{2}$. \cref{claim:8} tells us that
    \begin{equation}
        \expect_{\chi:\enc(x=1)=\chi} \sum_{\alpha}\big\| \big( X_{B} - (-1)^{\dec(\alpha)}\tau_{\theta,\phi}\big)\ket{\Psi_{\alpha|\chi}}\big\|^{2} \leq \delta_{0}/(\sin(2\theta)\tau_{\theta,\phi})^{2}, \label{eq:xrel}
    \end{equation}
    where $\tau_{\theta,\phi}= \frac{1}{\sin(2\theta)}(\Id - \cos(2\theta)Z_{B})$, which implies 
    \begin{equation*}
    \begin{aligned}
        8\| \big(\tilde{Z}_{B} - Z_{B}\big) X_{B}\ket{\Psi_{\alpha|\chi}}\big \|^{2} & \leq 16\| \big(\tilde{Z}_{B} - Z_{B}\big) \tau_{\theta,\phi}\ket{\Psi_{\alpha|\chi}}\big \|^{2} \\
        &+ 16\big\| \big(\tilde{Z}_{B} - Z_{B}\big)\big( X_{B} - (-1)^{\dec(\alpha)}\tau_{\theta,\phi}\big)\ket{\Psi_{\alpha|\chi}}\big\|^{2} \\
        &\leq 16\| \big(\tilde{Z}_{B} - Z_{B}\big) \tau_{\theta,\phi}\ket{\Psi_{\alpha|\chi}}\big \|^{2} \\
        &+ 16\big\| \tilde{Z}_{B} - Z_{B}\big\|_{\mathrm{op}}^{2}\big\|\big( X_{B} - (-1)^{\dec(\alpha)}\tau_{\theta,\phi}\big)\ket{\Psi_{\alpha|\chi}}\big\|^{2}.
    \end{aligned}
    \end{equation*}
    For the second term above, we have that $\big\| \tilde{Z}_{B} - Z_{B}\big\|_{\mathrm{op}}^{2} \leq 1 + 1/(2\cos^{2}(\phi))$, and when taking the expectation over $\chi:\enc(x=1)=\chi$ and summing over $\alpha$, we can apply \cref{eq:xrel} to obtain a contribution of $16(1 + 1/(2\cos^{2}(\phi)))\delta_{0}/(\sin(2\theta)\tau_{\theta,\phi})^{2}$. Finally, we bound the first term above using the fact that $\tilde{Z}_{B}$ and $Z_{B}$ commute, that is,
    \begin{equation*}
        \begin{aligned}
            16\| \big(\tilde{Z}_{B} - Z_{B}\big) \tau_{\theta,\phi}\ket{\Psi_{\alpha|\chi}}\big \|^{2} &=  16\big \|\tau_{\theta,\phi}\big(\tilde{Z}_{B} - Z_{B}\big) \ket{\Psi_{\alpha|\chi}}\big \|^{2} \\
            &\leq  16\|\tau_{\theta,\phi}  \|^{2}_{\mathrm{op}} \big\|\big(\tilde{Z}_{B} - Z_{B}\big) \ket{\Psi_{\alpha|\chi}}\big \|^{2}.
        \end{aligned}
    \end{equation*}
    The operator norm is bounded by $1/\sin^{2}(2\theta) + 1/(2\cos^{2}(\phi)\tan^{2}(2\theta))$, and when taking the expectation over $\chi:\enc(x=1)=\chi$ and summing over $\alpha$, we can apply \cref{lem:QHE1} followed by \cref{claim:5}, resulting in a contribution of $16(1/\sin^{2}(2\theta) + 1/(2\cos^{2}(\phi)\tan^{2}(2\theta)))(\delta_{0} + \negl(\lambda))$. By summing all the contributions from every use of the triangle inequality, we obtain the desired bound.
\end{proof}

\begin{claim}
    $\expect_{\chi:\enc(x=1)=\chi} \sum_{\alpha}\big \| (\tilde{X}_{B}P_{B}^{1} - P_{B}^{0}\tilde{X}_{B})\ket{\Psi_{\alpha|\chi}}\big\|^{2} \leq \delta_{6}/4$. \label{claim:11}
\end{claim}
\begin{proof}
    Note that 
    \begin{equation*}
        \tilde{X}_{B}P_{B}^{1} - P_{B}^{0}\tilde{X}_{B} = -\frac{1}{2}\{\tilde{Z}_{B},\tilde{X}_{B}\},
    \end{equation*}
    and hence the claim immediately follows from \cref{claim:10}.
\end{proof}
We are now ready to prove the self-testing statement on the state when $\enc(x=1) = \chi$:
\begin{lemma}
    Let $\widetilde{\qmodel}$ be any compiled model which satisfies $\pseudo[S_{\theta,\phi}] \geq \eta_{\theta,\phi}^{\mathrm{Q}} - \epsilon$, and let $V$ be the SWAP isometry defined in \cref{eq:extr}. Then there exists auxiliary states $\{\ket{\mathsf{aux}_{\alpha|\chi}}\}_{\alpha,\chi}$ such that
    \begin{equation}
        \expect_{\chi:\enc(x=1)=\chi} \sum_{\alpha}\big \| V\ket{\Psi_{\alpha|\chi}} - \big(\cos(\theta)\ket{0} + (-1)^{\dec(\alpha)}\sin(\theta)\ket{1}\big) \otimes \ket{\mathsf{aux}_{\alpha|\chi}} \big\|^{2} \leq \delta_{7},
    \end{equation}
    where $\delta_{7} := \delta_{6}/2 + 2\delta_{5}/(\sin^{2}(2\theta))$. \label{lem:st2}
\end{lemma}
\begin{proof}
    Consider the choice $\ket{\mathsf{aux}_{\alpha|\chi}} = \frac{1}{\cos(\theta)}P_{B}^{0}\ket{\Psi_{\alpha|\chi}}$. Then,
    \begin{equation*}
    \begin{aligned}
            \big \| & V\ket{\Psi_{\alpha|\chi}} - \big(\cos(\theta)\ket{0} + (-1)^{\dec(\alpha)}\sin(\theta)\ket{1}\big) \otimes \ket{\mathsf{aux}_{\alpha|\chi}} \big \|^{2} \\ &= \big \| \ket{1} \otimes \big( \tilde{X}_{B}P_{B}^{1} - (-1)^{\dec(\alpha)}\tan(\theta)P_{B}^{0}\big)\ket{\Psi_{\alpha|\chi}} \big \|^{2} \\
            &\leq 2\big \| \big( P_{B}^{0}\tilde{X}_{B} - (-1)^{\dec(\alpha)}\tan(\theta)P_{B}^{0}\big)\ket{\Psi_{\alpha|\chi}} \big \|^{2} + 2\big \| (\tilde{X}_{B}P_{B}^{1} - P_{B}^{0}\tilde{X}_{B})\ket{\Psi_{\alpha|\chi}}\big\|^{2}.
    \end{aligned}
    \end{equation*}
    Taking the expectation over $\chi:\enc(x=1)=\chi$ and summing over $\alpha$, we can bound the second term by $\delta_{6}/2$ using \cref{claim:11}. Using the triangle inequality again,
    \begin{equation*}
    \begin{aligned}
        2\big \| & \big( P_{B}^{0}\tilde{X}_{B} - (-1)^{\dec(\alpha)}\tan(\theta)P_{B}^{0}\big)\ket{\Psi_{\alpha|\chi}} \big \|^{2} \\ &\leq 4\big \| \big( (-1)^{\dec(\alpha)}P_{B}^{0}\tilde{T}_{\theta,\phi} - (-1)^{\dec(\alpha)}\tan(\theta)P_{B}^{0}\big)\ket{\Psi_{\alpha|\chi}} \big \|^{2} \\
        &+  4\big \| P_{B}^{0}\big(\tilde{X}_{B}-(-1)^{\dec(\alpha)}\tilde{T}_{\theta,\phi}  \big)\ket{\Psi_{\alpha|\chi}} \big \|^{2} \\
        &\leq 4\big \| \big( P_{B}^{0}\tilde{T}_{\theta,\phi} - \tan(\theta)P_{B}^{0}\big)\ket{\Psi_{\alpha|\chi}} \big \|^{2} \\
        &+ 4\big \| P_{B}^{0}\big \|_{\mathrm{op}}^{2}\big \| \big(\tilde{X}_{B}-(-1)^{\dec(\alpha)}\tilde{T}_{\theta,\phi}  \big)\ket{\Psi_{\alpha|\chi}} \big \|^{2},
    \end{aligned}
    \end{equation*}
    where $\tilde{T}_{\theta,\phi}= \frac{1}{\sin(2\theta)}(\Id - \cos(2\theta)\tilde{Z}_{B})$. Recall that $\big \| P_{B}^{0}\big \|_{\mathrm{op}}^{2} \leq 1/2$, and after taking the expectation over $\chi:\enc(x=1)=\chi$ and summing over $\alpha$, we can bound the second term using \cref{claim:9}, resulting in a contribution of $2\delta_{5}/(\sin^{2}(2\theta))$ to the final bound. Finally, by direct calculation we obtain
    \begin{equation}
        P_{B}^{0}\tilde{T}_{\theta,\phi} = \frac{1}{\sin(2\theta)}(P_{B}^{0} - \cos(2\theta)P_{B}^{0}\tilde{Z}_{B}) = \frac{1-\cos(2\theta)}{{\sin(2\theta)}}P_{B}^{0} = \tan(\theta)P_{B}^{0},
    \end{equation}
    which implies $\big \| \big( P_{B}^{0}\tilde{T}_{\theta,\phi} - \tan(\theta)P_{B}^{0}\big)\ket{\Psi_{\alpha|\chi}} \big \|^{2} = 0$. Summing the contributions from each application of the triangle inequality we arrive at the final bound. 
\end{proof}

\cref{lem:st1,lem:st2} complete the self-testing statement of the model without the action of measurements in the second round according to \cref{def:compst}.

\paragraph*{Extracting the measurements}

Next we turn our attention to the measurements. 
\begin{claim}
    $\expect_{\chi:\enc(x)=\chi} \sum_{\alpha}\big \| (V Z_{B} - (\sigma_{Z}\otimes \Id)V)\ket{\Psi_{\alpha|\chi}}\big\|^{2} \leq \delta_{8,x}$, where $\delta_{8,x} = \delta_{0} + x\cdot \negl(\lambda)$ for $x \in \{0,1\}$. \label{claim:12}
\end{claim}
\begin{proof}
    We first notice that
    \begin{equation*}
        \begin{aligned}
            (\sigma_{Z} \otimes \Id)V\ket{\Psi_{\alpha|\chi}} &= \sum_{b\in \{0,1\}} (-1)^{b}\ket{b} \otimes (\tilde{X}_{B})^{b}P_{B}^{b}\ket{\Psi_{\alpha|\chi}} \\
            &= \sum_{b\in \{0,1\}} \ket{b} \otimes (\tilde{X}_{B})^{b}P_{B}^{b}\tilde{Z}_{B}\ket{\Psi_{\alpha|\chi}},
        \end{aligned}
    \end{equation*}
    where we used the fact that $P_{B}^{b}\tilde{Z}_{B} = (-1)^{b}P_{B}^{b}$. Thus,
    \begin{equation*}
        \begin{aligned}
            \big \| (V Z_{B} - &(\sigma_{Z}\otimes \Id)V)\ket{\Psi_{\alpha|\chi}} \big \|^{2} = \big \| \sum_{b \in \{0,1\}} \ket{b} \otimes (\tilde{X}_{B})^{b}P_{B}^{b}(Z_{B} - \tilde{Z}_{B})\ket{\Psi_{\alpha|\chi}} \big \|^{2} \\
            &\leq \Bigg( \sum_{b \in \{0,1\}} \| (\tilde{X}_{B})^{b} P_{B}^{b}\|_{\mathrm{op}}^{2} \Bigg) \big \| (Z_{B} - \tilde{Z}_{B})\ket{\Psi_{\alpha|\chi}} \big \|^{2} \\
            &\leq  \| (Z_{B} - \tilde{Z}_{B})\ket{\Psi_{\alpha|\chi}} \big \|^{2},
        \end{aligned}
    \end{equation*}
    For the second inequality, we used the fact that $\|P_{B}^{b}\|_{\mathrm{op}}^{2} \leq 1/2$ and $\|\tilde{X}_{B} P_{B}^{1}\|_{\mathrm{op}} =\| P_{B}^{1}\|_{\mathrm{op}}$ since $\tilde{X}_{B}$ is unitary. For the case $\enc(x=0) = \chi$, taking the expectation, summing over $\alpha$ and applying \cref{claim:5} recovers the desired bound. For the case $\enc(x=1) = \chi$, the same procedure with an intermediate application of \cref{lem:QHE1}, which accumulates an extra factor of $\negl(\lambda)$, obtains the desired bound. 
\end{proof}

\begin{claim}
    $\expect_{\chi:\enc(x)=\chi} \sum_{\alpha}\big \| (V X_{B} - (\sigma_{X}\otimes \Id)V)\ket{\Psi_{\alpha|\chi}}\big\|^{2} \leq \delta_{9,x}$, where $\delta_{9,x} = 2\delta_{4}+ 2\negl(\lambda)\cdot (2-x) + \delta_{6}$. \label{claim:13}
\end{claim}
\begin{proof}
    By the triangle inequality
    \begin{equation*}
        \begin{aligned}
            \big \| (VX_{B} - (\sigma_{X}\otimes \Id)V)\ket{\Psi_{\alpha|\chi}}\big \|^{2} &\leq 2\big \| V(X_{B} - \tilde{X}_{B})\ket{\Psi_{\alpha|\chi}}\big \|^{2} \\
            &+ 2\big \| (V\tilde{X}_{B} - (\sigma_{X}\otimes \Id)V)\ket{\Psi_{\alpha|\chi}}\big \|^{2}.
        \end{aligned}
    \end{equation*}
    For the first term,
    \begin{equation*}
        2\big \| V(X_{B} - \tilde{X}_{B})\ket{\Psi_{\alpha|\chi}}\big \|^{2} \leq 2\Bigg( \sum_{b \in \{0,1\}} \| (\tilde{X}_{B})^{b} P_{B}^{b}\|_{\mathrm{op}}^{2} \Bigg) \big \| (X_{B} - \tilde{X}_{B})\ket{\Psi_{\alpha|\chi}} \big \|^{2},
    \end{equation*}
    where we performed the same calculation in \cref{claim:12}, which after taking the expectation over $\enc(x) = \chi$ and summing over $\alpha$, results in a contribution of $2\delta_{4}$ when $x = 0$ by \cref{claim:6}, and $2(\delta_{4}+\negl(\lambda))$ by \cref{lem:QHE1} and \cref{claim:6} when $x = 1$. 

    For the second term we explicitly calculate
    \begin{equation*}
        \begin{aligned}
            V\tilde{X}_{B} - (\sigma_{X} \otimes \Id)V = \ket{0} \otimes \big( P_{B}^{0} \tilde{X}_{B} - \tilde{X}_{B}P_{B}^{1}\big) + \ket{1} \otimes \big(\tilde{X}_{B}P_{B}^{1}\tilde{X}_{B} - P_{B}^{0}\big).
        \end{aligned}
    \end{equation*}
    Taking the norm squared, we end up with the sum of two terms. The first is given by
    \begin{equation*}
        2\big \| (P_{B}^{0}\tilde{X}_{B} - \tilde{X}_{B}P_{B}^{1})\ket{\Psi_{\alpha|\chi}} \big \|^{2},
    \end{equation*}
    which after taking the expectation over $\enc(x) = \chi$ and summing over $\alpha$ results in a contribution of $\delta_{6}/2$ for $x = 1$ by \cref{claim:11}, and $2(\delta_{6}/4 + \negl(\lambda))$ for $x = 0$ by \cref{lem:QHE1} and \cref{claim:11}. By the unitarity of $\tilde{X}_{B}$ the second is equal to $\big \| (P_{B}^{1}\tilde{X}_{B} - \tilde{X}_{B}P_{B}^{0})\ket{\Psi_{\alpha|\chi}} \big \|^{2}$. This is bounded by the same quantity as the previous, by noting that the proof of \cref{claim:11} is nearly identical with only the values $b$ of $P_{B}^{b}$ switched. 
\end{proof}
We are now ready to prove the final self-testing statement.

\begin{lemma}
    Let $\widetilde{\qmodel}$ be any compiled model which satisfies $\pseudo[S_{\theta,\phi}] \geq \eta_{\theta,\phi}^{\mathrm{Q}} - \epsilon$, $\widetilde{\qmodel}^{*}$ be the compiled-counterpart of model given in \cref{eq:honCHSH}, and let $V$ be the SWAP isometry defined in \cref{eq:extr}. Then there exists auxiliary states $\{\ket{\mathsf{aux}_{\alpha|\chi}}\}_{\alpha,\chi}$ such that for $x \in \{0,1\}$
    \begin{equation}
         \expect_{\chi:\enc(x)=\chi} \sum_{\alpha}\big \| VN_{b|y}\ket{\Psi_{\alpha|\chi}} - Q_{b|y}\ket{\varphi_{\dec(\alpha)|\chi}} \otimes \ket{\mathsf{aux}_{\alpha|\chi}} \big\|^{2} \leq \zeta_{x},
    \end{equation}
    where $\zeta_{x} = (\cos^{2}(\phi) \, \delta_{8,x} + \sin^{2}(\phi) \, \delta_{9,x})/2 + 2(1-x)\cdot \delta_{7} + 2x \cdot \delta_{0}$.\label{lem:measst}
\end{lemma}
\begin{proof}
    By definition $N_{b|y} = (1/2)(\Id + (-1)^{b}[\cos(\phi)\, Z_{B} + (-1)^{y} \sin(\phi) \, X_{B}])$. We also have that the target measurements are given by $Q_{b|y} = (1/2)(\Id + (-1)^{b}[\cos(\phi)\, \sigma_{Z} + (-1)^{y} \sin(\phi) \, \sigma_{X}])$ in \cref{eq:honCHSH}. We first calculate
    \begin{equation*}
        \begin{aligned}
            VN_{b|y}\ket{\Psi_{\alpha|\chi}} - (Q_{b|y} \otimes \Id)V\ket{\Psi_{\alpha|\chi}} &= \frac{(-1)^{b}\cos(\phi)}{2}(VZ_{B} - (\sigma_{Z} \otimes \Id)V)\ket{\Psi_{\alpha|\chi}} \\
            &+ \frac{(-1)^{b+y}\sin(\phi)}{2}(VX_{B} - (\sigma_{X} \otimes \Id)V)\ket{\Psi_{\alpha|\chi}}.
        \end{aligned}
    \end{equation*}
    Taking the vector norm of both sides, applying the triangle inequality, then taking the expectation over $\enc(x) = \chi$ and summing over $\alpha$, we arrive at
    \begin{equation}
        \expect_{\chi:\enc(x)=\chi} \sum_{\alpha} \big \| (VN_{b|y} - (Q_{b|y} \otimes \Id)V)\ket{\Psi_{\alpha|\chi}} \big \|^{2} \leq \frac{\cos^{2}(\phi) \, \delta_{8,x} + \sin^{2}(\phi) \, \delta_{9,x}}{4} \label{eq:bnd1}
    \end{equation}
    by \cref{claim:12} and \cref{claim:13}. Finally, we note that \cref{lem:st1,lem:st2} imply the existence of states $\{\ket{\mathsf{aux}_{\alpha|\chi}'}\}_{\alpha,\chi}$ such that 
    \begin{equation}
    \begin{aligned}
        &\expect_{\chi:\enc(x=0)=\chi} \sum_{\alpha}\big \| (Q_{b|y}\otimes \Id)V\ket{\Psi_{\alpha|\chi}} - Q_{b|y}\ket{\dec(\alpha)} \otimes \ket{\mathsf{aux}_{\alpha|\chi}'} \big\|^{2} \leq \delta_{0}, \\
        &\expect_{\chi:\enc(x=1)=\chi} \sum_{\alpha}\big \| (Q_{b|y}\otimes \Id)V\ket{\Psi_{\alpha|\chi}} - Q_{b|y}\big(\cos(\theta)\ket{0} + (-1)^{\dec(\alpha)}\sin(\theta)\ket{1}\big) \otimes \ket{\mathsf{aux}_{\alpha|\chi}'} \big\|^{2} \leq \delta_{7}, \label{eq:bnd2}
    \end{aligned}
    \end{equation}
    where we used that vector norm can only decrease under the application of a projector $Q_{b|y} \otimes \Id$. Note that for the honest strategy in \cref{eq:honCHSH}, the compiled counterpart consists of the sub-normalized states
    \begin{equation*}
        \begin{aligned}
            \ket{\varphi_{a|x=0}} &= 
            \begin{cases}
                \cos(\theta)\ket{0}_{B} \ \mathrm{if} \ a = 0, \\ 
                \sin(\theta)\ket{1}_{B} \ \mathrm{if} \ a = 1, 
            \end{cases} \\
            \ket{\varphi_{a|x=1}} &= \frac{\cos(\theta)\ket{0}_{B}  + (-1)^{a}\sin(\theta) \ket{1}_{B}}{\sqrt{2}}.
        \end{aligned}
    \end{equation*}
    For $\chi:\enc(x=0) = \chi$, define $\ket{\mathsf{aux}_{\alpha|\chi}} := \ket{\mathsf{aux}_{\alpha|\chi}'}/\cos(\theta)$ for $\alpha : \dec(\alpha) = 0$, and $\ket{\mathsf{aux}_{\alpha|\chi}} := \ket{\mathsf{aux}_{\alpha|\chi}'}/\sin(\theta)$ for $\alpha : \dec(\alpha) = 1$. For $\chi:\enc(x=1) = \chi$, define $\ket{\mathsf{aux}_{\alpha|\chi}} := \sqrt{2}\ket{\mathsf{aux}_{\alpha|\chi}'}$ for all $\alpha$. Then the triangle inequality implies
    \begin{equation*}
        \begin{aligned}
            \big \| VN_{b|y}\ket{\Psi_{\alpha|\chi}} - &Q_{b|y}\ket{\varphi_{\dec(\alpha)|x}} \otimes \ket{\mathsf{aux}_{\alpha|\chi}} \big\|^{2} 
            \leq 2\big \| VN_{b|y}\ket{\Psi_{\alpha|\chi}} - (Q_{b|y} \otimes \Id)V\ket{\Psi_{\alpha|\chi}} \big\|^{2} \\
            &+ 2\big \| (Q_{b|y} \otimes \Id)V\ket{\Psi_{\alpha|\chi}} - Q_{b|y}\ket{\varphi_{\dec(\alpha)|\chi}} \otimes \ket{\mathsf{aux}_{\alpha|\chi}} \big \|^{2},
        \end{aligned}
    \end{equation*}
    which is bounded by \cref{eq:bnd1,eq:bnd2}, after taking the expectation over $\enc(x) = \chi$ and summing over $\alpha$, by $\zeta_{x} = (\cos^{2}(\phi) \, \delta_{8,x} + \sin^{2}(\phi) \, \delta_{9,x}) + 2(1-x)\cdot \delta_{7} + 2x \cdot \delta_{0}$.
\end{proof}

\begin{proof}[Proof of \cref{thm:tilt_st}.] By noting that, for $x \in \{0,1\}$, the functions $\zeta_{x},\delta_0,$ and $\delta_7$ are of the form $f(\epsilon) + \negl(\lambda)''$, where $f(\epsilon) \geq 0$ and $f(0) = 0$, we see that \cref{lem:st1} and \cref{lem:st2} establish Eq. \eqref{eq:stdef_a}, while \cref{lem:measst} establishes Eq. \eqref{eq:stdef_b}. Together this shows that the generalized tilted-CHSH expressions are a self-test according to \cref{def:compst}, completing the proof.
\end{proof}

\newcommand{\etalchar}[1]{$^{#1}$}

\appendix

\section{Appendix}

\subsection{Efficient quantum circuits and homomorphic encryption}

To define a quantum homomorphic encryption scheme we require the following concepts from quantum cryptography.

\begin{definition}
    A procedure $\mcP$ is quantum polynomial time (QPT) if:
    \begin{enumerate}
        \item there exists a uniform logspace family of quantum circuits that implement $\mcP$, and
        \item the runtime of the circuit is polynomial in the number of qubits and the security parameter $\lambda\in \N$.
    \end{enumerate} \label{def:QPT}
    A family of quantum states $\mcF$ is QPT (preparable) if there is a QPT $\mcP$ for preparing $\mcF$.  
\end{definition}

We now define a quantum homomorphic encryption (QHE) scheme. A formal definition of QHE first appeared in \cite{BJ15}. We follow the description of QHE outlined in \cite{KLVY23,Bra18}:

\begin{definition}\label{def_qhe}
    A quantum homomorphic encryption scheme $\mathsf{Q}$ for a family of circuits $\mcC$ consists of a security parameter $\lambda\in \N$ and the following algorithms:
    \begin{enumerate}[(i)]
        \item A PPT algorithm $\gen$ which takes as input a unary encoding $1^\lambda$ of the security parameter $\lambda\in \N$ and outputs a secret key $\sk$.
        \item A PPT algorithm $\enc$ which takes as input the secret key $\sk$ and a plaintext $x\in \{0,1\}^n$ and produces a ciphertext $\chi\in \{0,1\}^k$.
        \item A QPT algorithm $\eval$ which takes as input a classical description of a quantum circuit $\mpC:\mcH\otimes (\C^2)^{\otimes n}\to (\C^2)^{\otimes m}$ from $\mcC$, a quantum plaintext $|\Psi\rangle\in \mcH$ on a Hilbert space, a ciphertext $\chi$, and evaluates a quantum circuit $\eval_\mpC(|\Psi\rangle\otimes |0\rangle^{\mathrm{poly}(\lambda,n)},\chi)$ producing a ciphertext $\alpha\in \{0,1\}^\ell$.
        \item A QPT algorithm $\dec$ which takes as input ciphertext $\alpha$, and secret key $\sk$, and produces a quantum state $|\Psi'\rangle$.
    \end{enumerate}
    \end{definition}
    
Although the existence of algorithms (i)-(iv) defines a QHE scheme, we consider several additional important properties a scheme may or may not possess:
\begin{enumerate}
    \item (Correctness with auxiliary input).
For every security parameter $\lambda\in \N$, secret key $\sk \leftarrow \gen(1^\lambda)$, classical circuit $\mpC:\mcH_A\otimes (\C^2)^{\otimes n}\to \{0,1\}^m$, quantum state $|\Psi\rangle_{AB}\in \mcH_A\otimes \mcH_B$, plaintext $x\in \{0,1\}^n$ ciphertext $\chi\leftarrow \enc(x,\sk)$, the following procedures produce states with negligible trace distance with respect to $\lambda$:
\begin{enumerate}
    \item Starting from the pair $(x,|\Psi\rangle_{AB})$, run the quantum circuit $\mpC$ on register $A$, outputting the classical string $a \in \{0,1\}^m$ along with the contents of register $B$.
    \item Starting from $(\chi,|\Psi\rangle_{AB})$, run the circuit $\eval_C(\cdot)$ on register $A$, obtaining ciphertext $\alpha\in \{0,1\}^\ell$, output $a'=\dec(\alpha,\sk)$ along with the contents of register $B$.
\end{enumerate}
\item (Security against efficient quantum distinguishers).
     Fix a secret key $\sk\leftarrow \gen(1^\lambda)$. Any quantum polynomial time adversary $\mfA$ with access to $\enc(\cdot,\sk)$ (but does not know $\sk$) cannot distinguish between ciphertexts $\chi\leftarrow \enc(x_0,\sk)$ and $\chi'\leftarrow \enc(x_1,\sk)$ with non-negligible probability in $\lambda$, where $x_0$ and $x_1$ are any plaintexts chosen by the adversary. That is
    \begin{equation*}
        |\Pr[\mfA^{\enc(x_0,\sk)}(x_0)=1]-\Pr[\mfA^{\enc(x_0,\sk)}(x_1)=1]|\leq \negl(\lambda),
    \end{equation*}
    for all pairs $(x_0,x_1)$.
\end{enumerate}

The KLVY compilation procedure requires schemes that satisfy (1) and (2). QHE schemes satisfying (1) and (2) have been described in \cite{Mah20, Bra18}.

\subsection{Additional proofs}\label{sec:app_b}

\begin{proof}[Proof of \cref{lem:proj_comp}]
    Let $V_{y}^{(\lambda)}:\mcH^{'(\lambda)} \to \mathbb{C}^{|\mcB|} \otimes \mcH^{'(\lambda)}$ be the isometry defined by
    \begin{equation}
        V_{y}^{(\lambda)}\ket{\phi} = \sum_{b \in \mcB} \ket{b} \otimes \sqrt{N_{b|y}^{'(\lambda)}}\ket{\phi}, \ \forall \ket{\phi} \in \mcH^{'(\lambda)}.
    \end{equation}
    Furthermore, let $U_{y}^{(\lambda)}$ be the unitary which satisfies $U_{y}^{(\lambda)}(\ket{0} \otimes \ket{\phi}) = V_y^{(\lambda)}\ket{\phi}$ for all $\ket{\phi} \in \mcH^{'(\lambda)}$. Define the projectors,
    \begin{equation}
        \tilde{N}_{b|y}^{(\lambda)} := U_{y}^{(\lambda)\dagger}(\ketbra{b}{b} \otimes \Id)U_{y}^{(\lambda)}.
    \end{equation}
    Since $|\mcB|$ is constant with respect to $\lambda$ and each $N_{b|y}^{'(\lambda)}$ is QPT, the resulting PVMs $\{\tilde{N}_{b|y}^{(\lambda)}\}_{b \in \mcB}$ are QPT for every $y \in \mcY$.
    For the sub-normalized states, let $\ket{\tilde{\Psi}_{\alpha|\chi}} \in \mcH^{'(\lambda)} \otimes \tilde{\mcH}^{(\lambda)}$ be any purification\footnote{Strictly speaking, since $\rho_{\alpha|\chi}^{(\lambda)}$ is sub-normalized, $\ket{\tilde{\Psi}_{\alpha|\chi}^{(\lambda)}}$ is equal to the purification of $\rho_{\alpha|\chi}^{(\lambda)} / \tr[\rho_{\alpha|\chi}^{(\lambda)}]$ weighted by $\tr[\rho_{\alpha|\chi}^{(\lambda)}]$, whenever $\tr[\rho_{\alpha|\chi}^{(\lambda)}] > 0$. } of $\rho_{\alpha|\chi}^{(\lambda)}$ with $\mcH^{'(\lambda)} \cong \tilde{\mcH}^{(\lambda)}$, and define 
    \begin{equation}
        \ket{\Psi_{\alpha|\chi}^{(\lambda)}} := \ket{0} \otimes \ket{\tilde{\Psi}_{\alpha|\chi}^{(\lambda)}} \in \mathbb{C}^{|\mcX|} \otimes \mcH^{'(\lambda)} \otimes \tilde{\mcH}^{(\lambda)} =: \mcH^{(\lambda)}.
    \end{equation}
    Again, since $|\mcX|$ is constant with respect to $\lambda$ the (sub-normalized) states $\ket{\Psi_{\alpha|\chi}^{(\lambda)}}$ are QPT-preparable. Now, extend each $\tilde{N}_{b|y}^{(\lambda)}$ to act trivially on the purifying system $\tilde{\mcH}^{(\lambda)}$ by defining $N_{b|y}^{(\lambda)} := \tilde{N}_{b|y}^{(\lambda)} \otimes \Id$. We observe
    \begin{equation}
    \begin{aligned}
        \tr\big[ N_{b|y}^{(\lambda)} \ketbra{\Psi_{\alpha|\chi}^{(\lambda)}}{\Psi_{\alpha|\chi}^{(\lambda)}}\big] &= \tr\Big[ \tilde{N}_{b|y}^{(\lambda)} \tr_{\tilde{Q}} [\ketbra{\Psi_{\alpha|\chi}^{(\lambda)}}{\psi_{\alpha|\chi}^{(\lambda)}}]\Big] \\
        &= \tr[\tilde{N}_{b|y}^{(\lambda)} (\ketbra{0}{0} \otimes \rho_{\alpha|\chi}^{(\lambda)})]\\
        &= \tr\big[\big(\ketbra{b}{b} \otimes \Id\big)U_{y}^{(\lambda)}\big( \ketbra{0}{0} \otimes \rho_{\alpha|\chi} \big)U_{y}^{(\lambda)\dagger}\big]\\
        &= \tr\Big[\sqrt{N_{b|y}^{'(\lambda)}}\rho_{\alpha|\chi}^{(\lambda)}\sqrt{N_{b|y}^{'(\lambda)}}\Big] = p^{(\lambda)}(\alpha,b|\chi,x),
    \end{aligned}
    \end{equation}
    where $\tilde{Q}$ denotes the purifying system $\tilde{\mcH}^{(\lambda)}$. Since $\tilde{\mcH}^{(\lambda)}$ has the same dimensions as $\mcH^{(\lambda)}$, the PVMs $N_{b|y}^{(\lambda)}$ are indeed QPT.
\end{proof}

\begin{proof}[Proof of \cref{thm:ExtendedSOS}]
    To begin, we write
    \begin{equation}
    \begin{aligned}
        \pseudo[P^{\dagger}P] &= \sum_{ij} \gamma_{i}^{*}\gamma_{j} \pseudo[(A_{x})^{k_{i}}w_{i}(B_{0},B_{1})(A_{x})^{k_{j}}w_{j}(B_{0},B_{1})] \\
        &= \sum_{ij} \gamma_{i}^{*}\gamma_{j} \pseudo[(A_{x})^{k_{i}+k_{j}}\bar{w}_{ij}(B_{0},B_{1})],
    \end{aligned}
    \end{equation}
    where we used the linearity of $\pseudo[\cdot]$ in the first line, and in the second line we used the fact that $\pseudo[w] = \pseudo[\bar{w}]$ (where $\bar{w}$ is the canonical form of the monomial $w$), and defined $\bar{w}_{ij}$ to be the canonical form of $w_{i}w_{j}$. We now need to consider two types of terms. First, when $k_{i} \oplus k_{j} = 0$, we apply the definition in \cref{eq:pseduo1} in conjunction with \cref{lem:QHE1} to write
    \begin{equation}
        \begin{aligned}
            &\sum_{ij:k_{i}\oplus k_{j} = 0}\gamma_{i}^{*}\gamma_{j}\pseudo[\bar{w}_{ij}(B_{0},B_{1})] \\ 
            &= \expect_{x' \in \mathcal{X}}\expect_{\chi:\enc(x')=\chi} \sum_{\alpha} \bra{\Psi_{\alpha|\chi}}\Bigg(\sum_{ij:k_{i}\oplus k_{j} = 0}\gamma_{i}^{*}\gamma_{j}\bar{w}_{ij}(B_{0},B_{1}) \Bigg)\ket{\Psi_{\alpha|\chi}} \\
            &\approx_{\negl(\lambda)} \expect_{\chi:\enc(x)=\chi} \sum_{\alpha} \bra{\Psi_{\alpha|\chi}}\Bigg(\sum_{ij:k_{i}\oplus k_{j} = 0}\gamma_{i}^{*}\gamma_{j}\bar{w}_{ij}(B_{0},B_{1}) \Bigg) \ket{\Psi_{\alpha|\chi}} \\
            &= \sum_{ij:k_{i}\oplus k_{j} = 0}\gamma_{i}^{*}\gamma_{j} \expect_{\chi:\enc(x)=\chi} \sum_{\alpha} \bra{\Psi_{\alpha|\chi}}\bar{w}_{ij}(B_{0},B_{1}) \ket{\Psi_{\alpha|\chi}}.
        \end{aligned}
    \end{equation}
    When $k_{i} \oplus k_{j} = 1$, we can apply \cref{eq:pseduo2} directly. Putting these two together, we observe
    \begin{equation}
    \begin{aligned}
        & \sum_{ij} \gamma_{i}^{*}\gamma_{j}\pseudo[(A_{x})^{k_{i}+k_{j}}\bar{w}_{ij}(B_{0},B_{1})]\\ &= \sum_{ij:k_{i}\oplus k_{j} = 0}\gamma_{i}^{*}\gamma_{j}\pseudo[\bar{w}_{ij}] + \sum_{ij:k_{i}\oplus k_{j} = 1}\gamma_{i}^{*}\gamma_{j}\pseudo[A_{x}\bar{w}_{ij}]\\ &\approx_{\negl(\lambda)} \sum_{ij:k_{i}\oplus k_{j} = 0}\gamma_{i}^{*}\gamma_{j} \expect_{\chi:\enc(x)=\chi} \sum_{\alpha} \bra{\Psi_{\alpha|\chi}}\bar{w}_{ij}(B_{0},B_{1}) \ket{\Psi_{\alpha|\chi}} \\ &+\sum_{ij:k_{i}\oplus k_{j} = 1} \gamma_{i}^{*}\gamma_{j} \expect_{\chi:\enc(x)=\chi} \sum_{\alpha}(-1)^{\dec(\alpha)}\bra{\Psi_{\alpha|\chi}} \bar{w}_{ij}(B_{0},B_{1}) \ket{\Psi_{\alpha|\chi}} \\
        &= \expect_{\chi:\enc(x)=\chi} \sum_{\alpha}\bra{\Psi_{\alpha|\chi}}\sum_{ij} (-1)^{\dec(\alpha) \cdot (k_{i} + k_{j})}\gamma_{i}^{*}\gamma_{j}w_{i}(B_{0},B_{1})w_{j}(B_{0},B_{1}) \ket{\Psi_{\alpha|\chi}} \\
        &= \expect_{\chi:\enc(x)=\chi} \sum_{\alpha}\bra{\Psi_{\alpha|\chi}}\Big|\sum_{i} (-1)^{\dec(\alpha) \cdot k_{i}}\gamma_{i}w_{i}(B_{0},B_{1})\Big|^{2}\ket{\Psi_{\alpha|\chi}} \geq 0. \label{eq:bigCalc}
    \end{aligned}
    \end{equation}
    Where in the fifth line, we used the fact that Bob's observables satisfy the canonical relations, so we can always replace the canonical monomial $\bar{w}_{ij}$ with $w_{i}w_{j}$. The final line is obtained by noting the square inside the expectation. We therefore conclude $\pseudo[P^{\dagger}P] \approx_{\negl(\lambda)} h$ for some $h \geq 0$, which implies $|\pseudo[P^{\dagger}P] - h| \leq \negl(\lambda)$, and $\pseudo[P^{\dagger}P] \geq h -   \negl(\lambda) \geq - \negl(\lambda)$ as required. Lastly, it is straightforward to verify from the definition that for any Bell functional $I$ we have $\pseudo(I)$ recovers the expected value under $\widetilde{\mathcal{\qmodel}}$.
\end{proof}

\end{document}